\documentclass[format=acmsmall, review=false, screen=true, natbib]{acmart}

\usepackage{booktabs}
\usepackage{amsmath}
\usepackage{subcaption}
\usepackage{url}
\usepackage{tikz}
\usepackage{graphbox}
\usepackage{adjustbox}
\usepackage{pgfplots}
\usepackage{enumitem}
\usepackage{placeins}
\usepackage{contour}

\usepackage{natbib}

\usetikzlibrary{fit,external,positioning}
\usetikzlibrary{shapes.geometric}
\usepackage[ruled, vlined, nofillcomment, linesnumbered]{algorithm2e}

\usepackage{ifthen}
\usepackage{xspace}

\newcommand{\dolphins}{\texttt{dolphins}\xspace}
\newcommand{\karate}{\texttt{karate}\xspace}
\newcommand{\lesmis}{\texttt{lesmis}\xspace}
\newcommand{\astro}{\texttt{astro}\xspace}
\newcommand{\enron}{\texttt{enron}\xspace}
\newcommand{\fb}{\texttt{fb1912}\xspace}
\newcommand{\hepph}{\texttt{hepph}\xspace}
\newcommand{\dblp}{\texttt{dblp}\xspace}
\newcommand{\gowalla}{\texttt{gowalla}\xspace}
\newcommand{\roadnet}{\texttt{roadnet}\xspace}
\newcommand{\skitter}{\texttt{skitter}\xspace}
\newcommand{\airports}{\texttt{airports}\xspace}
\newcommand{\trains}{\texttt{trains}\xspace}

\newcommand{\set}[1]{\left\{#1\right\}}
\newcommand{\pr}[1]{\left(#1\right)}
\newcommand{\fpr}[1]{\mathopen{}\left(#1\right)}

\newcommand{\abs}[1]{{\left|#1\right|}}

\newcommand{\enset}[2]{\left\{#1 ,\ldots , #2\right\}}

\newcommand{\real}{\mathbb{R}}
\newcommand{\np}{\textbf{NP}}

\newcommand{\funcdef}[3]{{#1}:{#2} \to {#3}}
\newcommand{\define}{\leftarrow}

\newcommand{\col}[1]{\mathcal{#1}}

\DeclareRobustCommand{\dispfunc}[2]{%
  \ensuremath{%
  \ifthenelse{\equal{#2}{}}%
    {\mathit{#1}}%
    {\mathit{#1}\fpr{#2}}}}

\newcommand{\cost}[1]{\ensuremath{\dispfunc{cost}{#1}}}
\newcommand{\bigo}[1]{\ensuremath{\dispfunc{\mathcal{O}}{#1}}}
\newcommand{\edges}[1]{\ensuremath{\dispfunc{E}{#1}}}
\newcommand{\crossedges}[1]{\ensuremath{\dispfunc{E_{\times}}{#1}}}
\newcommand{\marginaledges}[1]{\ensuremath{\dispfunc{E_{\Delta}}{#1}}}
\newcommand{\density}[1]{\ensuremath{\dispfunc{d}{#1}}}
\newcommand{\compact}[1]{\ensuremath{\dispfunc{f}{#1}}}
\newcommand{\compactgraph}[1]{\ensuremath{\dispfunc{F}{#1}}}
\newcommand{\prof}[1]{\ensuremath{\dispfunc{p}{#1}}}
\newcommand{\din}[1]{\ensuremath{\dispfunc{\mathrm{in}}{#1}}}
\newcommand{\dg}[1]{\ensuremath{\dispfunc{\operatorname{deg}}{#1}}}

\newcommand{\degree}[1]{\ensuremath{\mathrm{deg}({#1})}}

\newcommand{\adg}[1]{\ensuremath{\dispfunc{\operatorname{adg}}{#1}}}

\newcommand{\peel}{\textsc{GreedyLD}\xspace}
\newcommand{\core}{\textsc{Core}\xspace}
\newcommand{\decompose}{\textsc{ExactLD}\xspace}
\newcommand{\segment}{\textsc{Segment}\xspace}

\newtheorem{problem}{Problem}

\newcommand{\spara}[1]{{\smallskip\noindent{\bf {#1}}}}

\newcommand{\eparabegin}{\begin{description}}
\newcommand{\eparaend}{\end{description}}

\SetKwInOut{Output}{output}
\SetKwInOut{Input}{input}
\SetKw{Out}{output}

\pgfdeclarelayer{background}
\pgfdeclarelayer{foreground}
\pgfsetlayers{background,main,foreground}

\definecolor{yafaxiscolor}{rgb}{0.3, 0.3, 0.3}
\definecolor{yafcolor1}{rgb}{0.4, 0.165, 0.553}
\definecolor{yafcolor2}{rgb}{0.949, 0.482, 0.216}
\definecolor{yafcolor3}{rgb}{0.47, 0.549, 0.306}
\definecolor{yafcolor4}{rgb}{0.925, 0.165, 0.224}
\definecolor{yafcolor5}{rgb}{0.141, 0.345, 0.643}
\definecolor{yafcolor6}{rgb}{0.965, 0.933, 0.267}
\definecolor{yafcolor7}{rgb}{0.627, 0.118, 0.165}
\definecolor{yafcolor8}{rgb}{0.878, 0.475, 0.686}

\newlength{\yafaxispad}
\newlength{\yaftlpad}
\newlength{\yaflabelpad}
\newlength{\yafaxiswidth}
\newlength{\yafticklen}
\makeatletter
\def\pgfplots@drawtickgridlines@INSTALLCLIP@onorientedsurf#1{}
\makeatother

\newcommand{\yafdrawxaxis}[2]{
	\pgfplotstransformcoordinatex{#1}\let\xmincoord=\pgfmathresult 
	\pgfplotstransformcoordinatex{#2}\let\xmaxcoord=\pgfmathresult 
	\pgfsetlinewidth{\yafaxiswidth} 
	\pgfsetcolor{yafaxiscolor}
	\pgfpathmoveto{\pgfpointadd{\pgfpointadd{\pgfplotspointrelaxisxy{0}{0}}{\pgfqpointxy{\xmincoord}{0}}}{\pgfqpoint{-0.5\yafaxiswidth}{\yafaxispad}}}
	\pgfpathlineto{\pgfpointadd{\pgfpointadd{\pgfplotspointrelaxisxy{0}{0}}{\pgfqpointxy{\xmaxcoord}{0}}}{\pgfqpoint{0.5\yafaxiswidth}{\yafaxispad}}}
	\pgfusepath{stroke}

}
\newcommand{\yafdrawyaxis}[2]{
	\pgfplotstransformcoordinatey{#1}\let\ymincoord=\pgfmathresult 
	\pgfplotstransformcoordinatey{#2}\let\ymaxcoord=\pgfmathresult 
	\pgfsetlinewidth{\yafaxiswidth} 
	\pgfsetcolor{yafaxiscolor}
	\pgfpathmoveto{\pgfpointadd{\pgfpointadd{\pgfplotspointrelaxisxy{0}{0}}{\pgfqpointxy{0}{\ymincoord}}}{\pgfqpoint{\yafaxispad}{-0.5\yafaxiswidth}}}
	\pgfpathlineto{\pgfpointadd{\pgfpointadd{\pgfplotspointrelaxisxy{0}{0}}{\pgfqpointxy{0}{\ymaxcoord}}}{\pgfqpoint{\yafaxispad}{0.5\yafaxiswidth}}}
	\pgfusepath{stroke}
}

\newcommand{\yafdrawaxis}[4]{\yafdrawxaxis{#1}{#2}\yafdrawyaxis{#3}{#4}}

\pgfplotscreateplotcyclelist{yaf}{%
{yafcolor1,mark options={scale=0.75},mark=o}, 
{yafcolor2,mark options={scale=0.75},mark=square},
{yafcolor3,mark options={scale=0.75},mark=triangle},
{yafcolor4,mark options={scale=0.75},mark=o},
{yafcolor5,mark options={scale=0.75},mark=o},
{yafcolor6,mark options={scale=0.75},mark=o},
{yafcolor7,mark options={scale=0.75},mark=o},
{yafcolor8,mark options={scale=0.75},mark=o}} 

\pgfplotsset{axis y line=left, axis x line=bottom,
	tick align=outside,
	compat = 1.3,
	tickwidth=\yafticklen,
	clip = false,
	every axis title shift = 0pt,
    x axis line style= {-, line width = 0pt, opacity = 0},
    y axis line style= {-, line width = 0pt, opacity = 0},
    x tick style= {line width = \yafaxiswidth, color=yafaxiscolor, yshift = \yafaxispad},
    y tick style= {line width = \yafaxiswidth, color=yafaxiscolor, xshift = \yafaxispad},
    x tick label style = {font=\scriptsize, yshift = \yaftlpad},
    y tick label style = {font=\scriptsize, xshift = \yaftlpad},
    every axis y label/.style = {at = {(ticklabel cs:0.5)}, rotate=90, anchor=center, font=\scriptsize, yshift = -\yaflabelpad},
    every axis x label/.style = {at = {(ticklabel cs:0.5)}, anchor=center, font=\scriptsize, yshift = \yaflabelpad},
    x tick label style = {font=\scriptsize, yshift = 1pt},
    grid = major,
    major grid style  = {dash pattern = on 1pt off 3 pt},
	every axis plot post/.append style= {line width=\yafaxiswidth} ,
	legend cell align = left,
	legend style = {inner sep = 1pt, cells = {font=\scriptsize}},
	legend image code/.code={%
		\draw[mark repeat=2,mark phase=2,#1] 
		plot coordinates { (0cm,0cm) (0.15cm,0cm) (0.3cm,0cm) };%
	} 
}

\acmJournal{TKDD}
\acmVolume{1}
\acmNumber{1}
\acmArticle{1}
\acmYear{2017}
\acmMonth{1}
\copyrightyear{2009}

\setcopyright{acmlicensed}

\acmDOI{0000001.0000001}

\begin{document}

\title{Density-friendly Graph Decomposition}

\author{Nikolaj Tatti}
\affiliation{%
\institution{HIIT, University of Helsinki, Aalto University}
\city{Helsinki}
\country{Finland}}

\begin{abstract}
Decomposing a graph into a hierarchical structure via $k$-core
analysis is a standard operation in any modern graph-mining toolkit.
$k$-core decomposition is a simple and efficient method that allows
to analyze a graph beyond its mere degree distribution.
More specifically,
it is used to identify areas in the graph of increasing centrality and
connectedness, and it allows to reveal the structural organization of
the graph. 

Despite the fact that $k$-core analysis relies on vertex degrees, 
$k$-cores do not satisfy a certain, rather natural, density
property.
Simply put, the most central $k$-core is not necessarily the densest
subgraph. 
This inconsistency between $k$-cores and graph density provides the
basis of our study. 

We start by defining what it means for a subgraph to be 
{\em locally-dense},  and we show that our definition entails a 
nested chain decomposition of the graph, 
similar to the one given by $k$-cores, 
but in this case the components are arranged in order of increasing density. 
We show that such a {\em locally-dense decomposition} for a graph
$G=(V,E)$ can be computed in polynomial time. 
The running time of the exact decomposition algorithm is 
$\bigo{|V|^2|E|}$ but is significantly faster in practice.
In addition, we develop a linear-time algorithm
that provides a factor-2 approximation to the optimal locally-dense decomposition.
Furthermore, we show that the $k$-core decomposition is also a factor-2
approximation, however, as demonstrated by our experimental evaluation, 
in practice $k$-cores have different structure than locally-dense
subgraphs, and as predicted by the theory,  
$k$-cores are not always well-aligned with graph density.
\end{abstract}

\thanks{The research described in
this paper builds upon and extends the work appearing in \nobreak{WWW 2015} by \citet{tatti:2015:density}.}

\maketitle

\section{Introduction}
\label{sec:intro}

Finding dense subgraphs and communities is one of the most
well-studied problems in graph mining. 
Techniques for identifying dense subgraphs are used in a large number
of application domains, from biology, to web mining, to analysis of
social and information networks.
Among the many concepts that have been proposed for discovering dense
subgraphs, $k$-{\em cores} are particularly attractive for the
simplicity of their definition and the fact that they can be identified in
linear time.

The $k$-core of a graph is defined as a maximal subgraph in which every
vertex is connected to at least $k$ other vertices within that
subgraph. 
A $k$-{\em core decomposition} of a graph consists of finding the set
of all $k$-cores. 
A nice property is that the set of all $k$-cores 
forms a nested sequence of subgraphs, one included in the next. 
This makes the $k$-core decomposition of a graph a useful tool in
analyzing a graph by identifying areas of increasing centrality and
connectedness, and revealing the structural organization of the
graph.  
As a result, $k$-core decomposition has been applied to a number of
different applications, such as 
modeling of random graphs~\citep{bollobas1984evolution},
analysis of the internet topology~\citep{Carmi03072007},
social-network analysis~\citep{Seidman:1983tv}, 
bioinformatics~\citep{bader2003automated},
analysis of connection matrices of the human brain~\citep{hagmann08brain},
graph visualization~\citep{DBLP:journals/corr/abs-cs-0504107}, 
as well as
influence analysis~\citep{kitsak10influence,Ugander17042012} and team formation~\citep{Bonchi:2014kh}.

The fact that the $k$-core decomposition of a graph gives a chain of
subgraphs where vertex degrees are higher in the inner cores, suggests
that we should expect that the inner cores are, in certain sense, more
dense or more connected than the outer cores.
As we will show shortly, this statement is not true. 
Furthermore, in this paper we show how to obtain a graph decomposition
for which the statement is true, namely, the inner subgraphs of the
decomposition are denser than the outer ones. 
To quantify density, we adopt a classic notion
used in the densest-subgraph problem~\citep{Charikar:2000tg, Goldberg:1984up}, 
where density is defined as the ratio between the edges and the
vertices of a subgraph. 
This density definition can be also viewed as the average degree
divided by 2.

Our motivating observation is that $k$-cores are not ordered according
to this density definition.
The next example demonstrates that the most inner core is \emph{not}
necessarily the densest subgraph, and in fact, we can increase the
density by either adding or removing vertices.

\begin{example}\em
\label{ex:toy}
Consider the graph $G_1$ shown in Figure~\ref{fig:toy}, 
consisting of 6 vertices and 9 edges. 
The density of the whole graph is $9/6 = 1.5$.
The graph has three $k$-cores:
a $3$-core marked as $C_1$, a 
$2$-core marked as $C_2$, and 
a $1$-core, corresponding the the whole graph and marked as $C_3$.
The core $C_1$ has density $6/4 = 1.5$ 
(it contains $6$ edges and $4$ vertices),
while the core $C_2$ has density $8/5 = 1.6$
(it contains $8$ edges and $5$ vertices).
In other words, $C_1$ has lower density than $C_2$,
despite being an inner core.

Let us now consider $G_2$ shown in Figure~\ref{fig:toy}.  
This graph has a single core, 
namely a $2$-core, containing the whole graph.  
The density of this core is equal to $11/8 = 1.375$. 
However, a subgraph $B_1$ contains $7$ edges and $5$ vertices,
giving us density $7/5 = 1.4$, 
which is higher than the density of the only core.
\end{example}

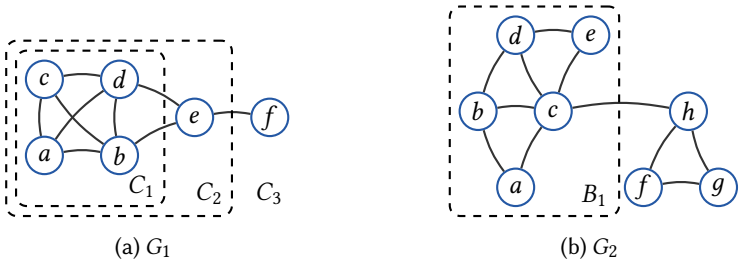
\begin{figure}[ht]
\tikzstyle{exnode} = [thick, draw = yafcolor5, circle, inner sep = 1pt, text=black, minimum width=14pt]
\tikzstyle{exedge} = [black!80, thick]
\tikzstyle{exbox} = [dashed, black, thick, draw, rounded corners]
\hspace*{\fill}
\subcaptionbox{$G_1$}{
\begin{tikzpicture}
\node[exnode] (n1) at (0, 0) {$a$};
\node[exnode] (n2) at (1, 0) {$b$};
\node[exnode] (n3) at (0, 1) {$c$};
\node[exnode] (n4) at (1, 1) {$d$};
\node[exnode] (n5) at (2, 0.5) {$e$};
\node[exnode] (n6) at (3, 0.5) {$f$};

\node[inner sep = 0pt] (l1) at (1.3, -0.4) {$C_1$};
\node[inner sep = 0pt] (l2) at (2.2, -0.5) {$C_2$};
\node[inner sep = 0pt] (l3) at (3, -0.5) {$C_3$};

\draw[exedge, bend left = 10] (n1) edge (n2);
\draw[exedge, bend left = 10] (n1) edge (n3);
\draw[exedge, bend left = 10] (n1) edge (n4);
\draw[exedge, bend left = 10] (n2) edge (n3);
\draw[exedge, bend left = 10] (n2) edge (n4);
\draw[exedge, bend left = 10] (n3) edge (n4);

\draw[exedge, bend left = 10] (n2) edge (n5);
\draw[exedge, bend left = 10] (n4) edge (n5);
\draw[exedge, bend left = 10] (n5) edge (n6);

\node[exbox, fit = (n1) (n2) (n3) (n4) (l1)] (b1) {};
\node[exbox, fit = (b1) (n5) (n3) (n4) (l2)] {};
\end{tikzpicture}
}\hfill
\subcaptionbox{$G_2$}{
\begin{tikzpicture}
\node[exnode] (n1) at (0.5, -1) {$a$};
\node[exnode] (n2) at (0, 0) {$b$};
\node[exnode] (n3) at (1, 0) {$c$};
\node[exnode] (n4) at (0.5, 1) {$d$};
\node[exnode] (n8) at (1.5, 1) {$e$};

\node[exnode] (n5) at (2.2, -1) {$f$};
\node[exnode] (n6) at (3.2, -1) {$g$};
\node[exnode] (n7) at (2.8, 0) {$h$};

\node[inner sep = 0pt] (l1) at (1.55, -1.1) {$B_1$};

\draw[exedge, bend left = 10] (n1) edge (n2);
\draw[exedge, bend left = 10] (n1) edge (n3);
\draw[exedge, bend left = 10] (n2) edge (n3);
\draw[exedge, bend left = 10] (n2) edge (n4);
\draw[exedge, bend left = 10] (n3) edge (n4);
\draw[exedge, bend left = 10] (n3) edge (n8);
\draw[exedge, bend left = 10] (n4) edge (n8);

\draw[exedge, bend left = 10] (n3) edge (n7);

\draw[exedge, bend left = 10] (n6) edge (n7);
\draw[exedge, bend left = 10] (n5) edge (n7);
\draw[exedge, bend left = 10] (n5) edge (n6);

\node[exbox, fit = (n1) (n2) (n3) (n4) (l1) (n8)] (b1) {};
\end{tikzpicture}
}\hspace*{\fill}
\caption{Toy graphs used in Example~\ref{ex:toy}.}
\label{fig:toy}
\end{figure}

This example motivates us to define an alternative, 
more density-friendly, graph decomposition, which we call 
\emph{locally-dense decomposition}. 
We are interested in a decomposition such that 
($i$) the density of the inner subgraphs is higher than the density of
the outer subgraphs, 
($ii$) the most inner subgraph corresponds to the densest subgraph,
and 
($iii$) we can compute or approximate the decomposition efficiently.

We achieve our goals by first defining a \emph{locally-dense} subgraph,
essentially a subgraph whose density cannot be improved by adding and
deleting vertices. 
We show that these subgraphs are arranged into a hierarchy such
that the density decreases as we go towards outer subgraphs and that
the most inner subgraph is in fact the densest subgraph.

We provide two efficient algorithms to discover this hierarchy. 
The first algorithm extends the exact algorithm for discovering the
densest subgraph given by~\citet{Goldberg:1984up}. 
This algorithm is based on solving a minimum cut
problem on a certain graph that depends on a parameter $\alpha$.  
Goldberg showed that for a certain value $\alpha$
(which can be found by binary search), 
the minimum cut recovers the densest subgraph. 
One of our contributions is to shed more light into Goldberg's
algorithm and show that the same construction allows to discover
\emph{all} locally-dense subgraphs by varying $\alpha$.

Our second algorithm extends the linear-time algorithm by~\citet{Charikar:2000tg} for
approximating dense subgraphs.
This algorithm first orders vertices by deleting iteratively a vertex
with the smallest degree, and then selects the densest subgraph
respecting the order. 
We extend this idea by using the same order, and finding first the
densest subgraph respecting the order, and then iteratively finding
the second densest subgraph containing the first subgraph, and so on. 
We show that this algorithm can be executed in linear time and it
achieves a factor-$2$ approximation guarantee.

Charikar's algorithm and the algorithm for discovering a $k$-core
decomposition are very similar: they both order vertices by deleting vertices with the
smallest degree. 
We show that this connection is profoundly deep and
we demonstrate that a $k$-core decomposition provides a
\mbox{factor-$2$} approximation for locally-dense decomposition.
On the other hand, our experimental evaluation shows that 
in practice $k$-cores have different structure than locally-dense
subgraphs, and as predicted by the theory,  
$k$-cores are not always well-aligned with graph density.

It is possible that the decomposition results a significant amount of subgraphs.
In such a case it may be useful to constraint the number of the subgraphs.
We approach this problem by defining an optimization criterion for a segmentation of $k$
nested subgraphs. The objective function will be based on a statistical model.
We will show that to optimize this particular objective, we need to
(\emph{i}) find locally-dense subgraphs, and
(\emph{ii}) reduce the number with a dynamic program.
We also show that if we replace the first step with the greedy algorithm,
then the resulting algorithm yields a factor-2 approximation guarantee.

The remainder of paper is organized as follows.  We give preliminary notation
in Section~\ref{sec:prel}.  We introduce the locally-dense subgraphs in
Section~\ref{sec:monotonic}, present algorithms for discovering the subgraphs
in Section~\ref{sec:discovery}, and describe the connection to $k$-core
decomposition in Section~\ref{sec:core}. 
We introduce the constrained version of the problem in Section~\ref{sec:segmentation}.
We present the related work in
Section~\ref{sec:related} and present the experiments in
Section~\ref{sec:experiments}.
Finally, we conclude the paper with discussion in Section~\ref{sec:conclusions}.

\section{Preliminaries}
\label{sec:prel}

\spara{Graph density.}
Let $G=(V,E)$ be a graph with $|V|=n$ vertices and $|E|=m$ edges. 
Given a subset of vertices $X\subseteq V$, 
it is common to define $\edges{X}=\set{(x,y)\in E \mid x,y \in X}$, that is,  
the edges of $G$ that have both end-points in $X$. 
The \emph{density} of the vertex set $X$ is then defined to be
\[
	\density{X} = \frac{\abs{E(X)}}{\abs{X}}, 
\]
that is, half of the \emph{average degree} of the subgraph induced by $X$.
The set of vertices $X\subseteq V$ that maximizes the density measure 
$\density{X}$ is the \emph{densest subgraph} of $G$.\footnote{We should point
out that density is also often defined as $\abs{E(X)} / {\abs{X} \choose 2}$.
This is not the case for this paper.}

The problem of finding the densest subgraph can be solved in
polynomial time. 
A very elegant solution that involves a mapping to a series of
minimum-cut problems was given by~\citet{Goldberg:1984up}.
As the fastest algorithm to solve the minimum-cut problem runs in 
$\bigo{mn}$ time, this approach is not scalable to very large graphs. 
On the other hand, there exists a linear-time algorithm that provides a
factor-$2$ approximation to the densest-subgraph problem~\citep{Asahiro:1996uq,Charikar:2000tg}.
This is a greedy algorithm, which starts with the input graph, and
iteratively removes the vertex with the lowest degree, until left
with an empty graph. Among all subgraphs considered during this
vertex-removal process, the algorithm returns the densest.

Next we will provide graph-density definitions that relate pairs of vertex
sets. Given two non-overlapping sets of vertices $X$ and $Y$ we first define the  
{\em cross  edges} between $X$ and $Y$ as
\[
	\crossedges{X, Y} = \set{(x, y) \in E \mid x \in X, y \in Y}\quad.
\]
We then define the {\em marginal edges} from $X$ with respect to $Y$.
Those are the edges that have one end-point in $X$ and the other
end-point in either $X$ or~$Y$, that is,
\[
	\marginaledges{X,Y}=\edges{X} \cup \crossedges{X,Y}\quad.
\]
The set $\marginaledges{X,Y}$ represents the additional edges that
will be included in the induced subgraph of $Y$ if we expand $Y$ by
adding~$X$.

Assume that $X$ and $Y$ are non-overlapping.
Then, we define the \emph{outer density} of $X$ with
respect to $Y$ as
\[
	\density{X, Y} = \frac{\abs{\marginaledges{X, Y}}}{\abs{X}}\quad.
\]
That is, these are the extra edges, on average, that we bring to $Y$ if we expand
it by appending $X$.

Now that we have defined a special case when $X$ and $Y$ are disjoint, we can now consider
a more general case, that is, when $X$ and $Y$ are overlapping. Here 
we would be interested in the outer density of vertices in $X$ that are \emph{not already
included} in~$Y$. Hence, we will expand the definition of outer density to a more general case by defining
\[
	\density{X, Y} = \density{X \setminus Y, Y} = \frac{\abs{\marginaledges{X \setminus Y, Y}}}{\abs{X \setminus Y}}\quad.
\]

\spara{$\mathbf{k}$-cores.}
We briefly review the basic background regarding $k$-cores.
The concept 
was introduced by~\citet{Seidman:1983tv}.
 
Given a graph $G=(V,E)$, 
a set of vertices $X\subseteq V$ is a $k$-core if 
every vertex in the subgraph induced by $X$ has degree at least $k$,
and $X$ is maximal with respect to this property. 
A $k$-core of $G$ can be obtained by recursively removing all the vertices
of degree less than $k$, until all vertices in the remaining graph
have degree at least $k$.

It is not hard to see that 
if $\set{C_i}$ is the set of all distinct $k$-cores of $G$
then $\set{C_i}$ forms a nested chain
\[
	\emptyset = C_0\subsetneq C_1 \subsetneq \cdots \subsetneq C_\ell = V\quad.
\]
Furthermore, the set of vertices $S_k$ that belong in a $k$-core but
not in a $(k-1)$-core is called $k$-\emph{shell}.

The $k$-\emph{core decomposition} of $G$ is the process of identifying
all $k$-cores (and all $k$-shells).
Therefore, the $k$-core decomposition of a graph identifies
progressively the internal cores and decomposes the graph shell by
shell. 
A linear-time algorithm to obtain the $k$-core decomposition 
was given by~\citet{Matula1983smallest}.
The algorithm starts by provisionally assigning each vertex $v$ to a
core of index $\degree{v}$, an upper bound to the correct core of a vertex.
It then repeatedly removes the vertex with the smallest degree, and
updates the core index of the neighbors of the removed vertex. 
Note the similarity of this algorithm, with the $2$-approximation
algorithm for the densest-subgraph problem~\citep{Charikar:2000tg}.

\section{Locally-dense graph decomposition}
\label{sec:monotonic}

In this section we present the main concept introduced in this paper,
the {\em locally-dense decomposition} of a graph.
We also discuss the properties of this decomposition.
We start by defining the concept of a {\em locally-dense subgraph}.
\begin{definition}
A set of vertices $W$ is \emph{locally dense} if there are no 
$X \subseteq W$ and $Y$ satisfying
$Y \cap W = \emptyset$ such that
\[
	\density{X, W \setminus X} \le \density{Y, W}\quad.
\]
\end{definition}
In other words, 
for $W$ to be locally dense 
there should not be an $X$ ``inside'' $W$ and a $Y$ ``outside'' $W$ 
so that the density that $Y$ brings to $W$ is larger than the density
that $X$ brings.

Due to the notational simplicity, we will often refer to these sets of vertices as
subgraphs.

Interestingly, the property of being locally dense induces a nested chain of
subgraphs in $G$.
\begin{proposition}
Let $U$ and $W$ be locally-dense subgraphs. 
Then either $U \subseteq W$ or $W \subseteq U$.
\end{proposition}

\begin{proof}
Assume otherwise. 
Define $X = U \setminus W$ and 
$Y = W \setminus U$. 
Both $X$ and $Y$ should be non-empty sets. 
Then either 
$\density{X, U \cap W} \leq \density{Y, U \cap W}$ or 
$\density{X, U \cap W} > \density{Y, U \cap W}$.
Assume the former. This implies
\begin{eqnarray*}
	\density{X, U \setminus X} = \density{X, U \cap W} \leq  \density{Y, U \cap W} \leq  \density{Y, U},
\end{eqnarray*}
which contradicts the fact that $U$ is locally dense. 
For the first equality we used the fact that 
$U \setminus X = U \cap W$,
while for the last inequality we used the fact that 
$\crossedges{Y, U \cap W} \leq \crossedges{Y, U}$.

The case 
$\density{X, U \cap W} > \density{Y, U \cap  W}$ 
is similar.
\end{proof}

The proposition implies that the set of locally-dense subgraphs of a
graph forms a nested chain, 
in the same way that the set of $k$-cores does.

\begin{corollary}
\label{corollary:chain}
A set of locally-dense subgraphs can be arranged into a sequence 
$B_0 \subsetneq B_1 \subsetneq \cdots \subsetneq B_k$,
where $k \leq \abs{V}$.  Moreover, $\density{B_{i}, B_{i - 1}} > \density{B_{i + 1}, B_{i}}$ for $1 \leq i < k$.
\end{corollary}

The chain of locally-dense subgraphs of a graph $G$, 
as specified by  Corollary~\ref{corollary:chain}, 
defines the {\em locally-dense decomposition} of~$G$.

\begin{example}
The locally-dense composition of $G_1$ given in Figure~\ref{fig:toy}
is $\emptyset \subsetneq C_2 \subsetneq C_3 = V$,
This is the $k$-core decomposition without $C_1$.
The locally-dense composition of $G_2$ given in Figure~\ref{fig:toy}
is $\emptyset \subsetneq B_1 \subsetneq V$. Note that both $C_2$ and $B_1$
are the densest subgraphs in their respective graphs.
\end{example}

We proceed to characterize the locally-dense subgraphs of the
decomposition with respect to their {\em global} density in the 
whole graph $G$. 
We want to characterize the global density of subgraph $B_i$ of the
decomposition. 
$B_i$ cannot be denser than the previous subgraph $B_{i-1}$ in the
decomposition, however, we want to measure the density that the
additional vertices $S_i=B_i\setminus B_{i-1}$ bring.
This density involves edges among vertices of $S_i$ and edges from
$S_i$ to the previous subgraph $B_{i-1}$.
This is captured precisely by the concept of {\em outer density}
$\density{B_i,B_{i-1}}$ defined in the previous section.
As the following proposition shows
the outer density of $B_i$ with respect to $B_{i-1}$ is maximized over
all subgraphs that contain $B_{i-1}$.
In other words, 
$B_i$ is the densest subgraph we can choose after $B_{i-1}$, 
given the containment constraint.

\begin{proposition}
\label{prop:maximal}
Let $\set{B_i}$ be the chain of locally-dense subgraphs.
Then $B_0 = \emptyset$, $B_k = V$, and $B_{i}$ is the densest subgraph properly containing $B_{i - 1}$,
\[
	B_{i} = \arg \max_{W \supsetneq B_{i - 1}} \density{W, B_{i - 1}}\quad.
\]
\end{proposition}

To prove the proposition we will use the following three lemmas.

\begin{lemma}
\label{lem:add}
Let $X \subseteq Y$ be two sets of vertices with $Y \neq \emptyset$. Assume a third non-empty set $Z$ with $Z \cap Y = \emptyset$.
Then one of the following three cases follows:
\begin{itemize}
\item $\density{Z, Y}  >  \density{Y \cup Z, X}  >  \density{Y, X}$, or
\item $\density{Z, Y}  <  \density{Y \cup Z, X}  <  \density{Y, X}$, or
\item $\density{Z, Y}  =  \density{Y \cup Z, X}  =  \density{Y, X}$.
\end{itemize}
\end{lemma}

\begin{proof}
Write $\alpha = \frac{\abs{Y}}{\abs{Y} + \abs{Z}}$.
We can rewrite $\density{Y \cup Z, X}$ as
\[
	\density{Y \cup Z, X}
	= \frac{\abs{\marginaledges{Y \cup Z, X}}}{\abs{Y} + \abs{Z}}
	= \frac{\abs{\marginaledges{Y, X}} + \abs{\marginaledges{Z, Y}}}{\abs{Y} + \abs{Z}}
	= \alpha\density{Y, X} + (1 - \alpha)\density{Z, Y}\quad.
\]
This shows that
either
$\density{Z, Y}  \geq  \density{Y \cup Z, X}  \geq  \density{Y, X}$
or
$\density{Z, Y}  \leq  \density{Y \cup Z, X}  \leq  \density{Y, X}$.
Since $0 < \alpha < 1$ it follows that $\density{Z, Y}  =  \density{Y \cup Z, X}$ if and only if $\density{Y \cup Z, X}  =  \density{Y, X}$.
The three cases follows.
\end{proof}

Let $C_i$ be the sequence defined as $C_i = \arg \max_{W \supsetneq C_{i - 1}} \density{W}$, in case of a tie, select a larger graph, and $C_0 = \emptyset$.

\begin{lemma}
\label{lem:maxmonotone}
$\density{C_j, C_{j - 1}} > \density{C_i, C_{i - 1}}$ for $j < i$.
\end{lemma}

\begin{proof}
We only need to show that the lemma holds $j = i - 1$.
Assume otherwise: $\density{C_i, C_{i - 1}} \geq \density{C_{i - 1}, C_{i - 2}}$.

Write $Z = C_i \setminus C_{i - 1}$, $Y = C_{i - 1}$, and $X = C_{i - 2}$.
Since $\density{Z, Y} = \density{C_i, C_{i - 1}} $,
Lemma~\ref{lem:add} implies that
\[
	\density{C_i, C_{i - 2}} =  \density{Y \cup Z, X} \geq  \density{Y, X} = \density{C_{i - 1}, C_{i - 2}},
\]
violating the optimality of $C_{i - 1}$.
\end{proof}

\begin{lemma}
\label{lem:maxdel}
If $Z \subseteq C_j \setminus C_{j - 1}$ and $Z \neq \emptyset$, then $\density{Z, C_j \setminus Z} \geq \density{C_j, C_{j - 1}}$.
\end{lemma}

\begin{proof}
Assume otherwise: $\density{Z, C_j \setminus Z} < \density{C_j, C_{j - 1}}$.
Write $X = C_{j - 1}$, $Y = C_j \setminus Z$. Lemma~\ref{lem:add} implies that
\[
	\density{C_j, C_{j - 1}} = \density{Z \cup Y, X} < \density{Y, X} = \density{C_j \setminus Z, C_{j - 1}},
\]
violating the optimality of $C_j$.
\end{proof}

\begin{proof}[Proof of Proposition~\ref{prop:maximal}]
We need to show that $C_i = B_i$.
Fix $i$ and assume inductively that $C_j = B_j$ for all $j < i$.

We will first show that $C_i$ is locally dense: we argue that there are no sets
$X$ and $Y$ with $X \subseteq C_i$ and $Y \cap C_i = \emptyset$ that can serve
as certificates for $C_i$ being non locally-dense.

Fix any $X \subseteq C_i$. 
Define $X_j = X \cap (C_j \setminus C_{j - 1})$ and
$U_j = (C_i \setminus X) \cup C_{j - 1}$ for $j \leq i$.

We claim that $C_j \subseteq U_j \cup X_j$.
Let $x \in C_j$. If $x \in C_{j - 1}$, then $x \in U_j$.
Assume that $x \in C_j \setminus C_{j - 1}$.
If $x \in X$, then $x \in X_j$.
If $x \notin X$, then $x \in C_j \setminus X \subseteq C_i \setminus X \subseteq U_j$.
Thus, $C_j \subseteq U_j \cup X_j$, which in turns implies
that $C_j \setminus X_j \subseteq U_j \setminus X_j$.

This leads to
\begin{flalign*}
	&& \density{X_j, U_j \setminus X_j} 
	& \geq \density{X_j, C_j \setminus X_j} & (C_j \setminus X_j \subseteq U_j \setminus X_j)\\
	&& &\geq \density{C_j, C_{j - 1}}  & \text{(Lemma~\ref{lem:maxdel})}\\
	&& & > \density{C_i, C_{i - 1}}\quad.  & \text{(Lemma~\ref{lem:maxmonotone})}\\
\end{flalign*}

This inequality leads to
\begin{eqnarray*}
\density{X, C_i \setminus X} & = & \sum_{j = 1, X_j \neq \emptyset}^i \frac{\abs{X_j}}{\abs{X}} \density{X_j, U_j \setminus X_j} \\
& \geq & \sum_{j = 1}^i \frac{\abs{X_j}}{\abs{X}} \density{C_i, C_{i - 1}} \\
&  = & \density{C_i, C_{i - 1}}\quad.
\end{eqnarray*}

Consider also any set $Y$ with $Y \cap C_i = \emptyset$. 
Due to the optimality of $C_i$ and Lemma~\ref{lem:add} we must have $\density{Y, C_i} < \density{C_i, C_{i - 1}}$.

We conclude that for any $X$ and $Y$ with 
$X\subseteq U$ and $Y\cap C_i=\emptyset$ 
it is $\density{X, C_i \setminus X}>\density{Y, C_i}$, 
which shows that $C_i$ is locally dense.

Now, we can safely assume $C_i = B_j$ for some $j$.  
We need to show that $j = i$.
By induction we know that $C_{i - 1} = B_{i - 1}$.
This guarantees that $j \geq i$.
Assume $j > i$. 
Since $C_i$ is maximal, we have $\density{B_j \setminus B_i, B_i} < \density{B_i, B_{i-1}}$.

Since $B_i$ is locally-dense, we have $\density{B_j \setminus B_i, B_i} < \density{B_i, B_{i-1}}$.
Lemma~\ref{lem:add} now implies that
\[
	\density{B_j, B_{i - 1}} < \density{B_i, B_{i - 1}}
\]
which contradicts the optimality of $C_i = B_j$. Thus $i = j$.
\end{proof}

As a consequence of the previous proposition we can characterize the first
subgraph in the decomposition.

\begin{corollary}
Let $\set{B_i}$ be a locally-dense decomposition of a graph $G$.
Then $B_1$ is the densest subgraph of $G$.
\end{corollary}

The above discussion motivates the problem of 
locally-dense graph decomposition, 
which is the focus of this paper.

\begin{problem}
\label{problem:LDGD}
Given a graph $G=(V,E)$
find a maximal sequence of locally-dense subgraphs
\[
\emptyset = B_0 \subsetneq B_1 \subsetneq \cdots \subsetneq B_k=V\quad.
\] 
\end{problem}

\section{Decomposition algorithms}
\label{sec:discovery}

In this section we propose two algorithms for the problem of
locally-dense graph decomposition (Problem~\ref{problem:LDGD}).
The first algorithm gives an {\em exact solution}, and runs in
worst-case time $\bigo{n^2m}$,
but it is significantly faster in practice. 
The second algorithm is a linear-time algorithm that provides a
factor-$2$ approximation guarantee. 

Both algorithms are inspired by corresponding algorithms for the
densest-subgraph problem. 
The first algorithm by the exact algorithm of~\citet{Goldberg:1984up},
 and the second algorithm by the greedy linear-time algorithm of~\citet{Charikar:2000tg}. 

\subsection{Exact algorithm}

We start our discussion on the exact algorithm for locally-dense graph
decomposition by reviewing Goldberg's algorithm~\citep{Goldberg:1984up}
for the densest-subgraph problem.

Recall that the densest-subgraph problem asks to find the subset of
vertices $W$ that maximizes
$\density{W} = \abs{E(W)}/\abs{W}$.
Given a graph $G=(V,E)$ and a positive number $\alpha \geq 0$
define a function
\[
	\compact{\alpha} = \max_{W\subseteq V} \set{ \abs{E(W)} - \alpha \abs{W}},
\]
and the maximizer 
\[
	\compactgraph{\alpha} = \arg \max_{W\subseteq V} \set{ \abs{E(W)} - \alpha \abs{W} },
\]
where ties are resolved by picking the largest $W$.
Note that $\compact{}$ 
decreases as $\alpha$ increases,  and as $\alpha$ exceeds a certain
value, $\compact{}$ becomes $0$ by taking $W=\emptyset$. 
Goldberg observed that the densest-subgraph problem is equivalent to
the problem of finding the largest value of $\alpha^*$ for which 
the maximizer set 
$\compactgraph{\alpha^*}=W^*$ is non empty.\footnote{This 
observation is an instance of 
{\em fractional programming}~\citep{dinkelbach67fractional}.}
The densest subgraph is precisely this maximizer set $W^*$.
Furthermore, Goldberg showed how to find the vertex set
$W=\compactgraph{\alpha}$, for a given value of $\alpha$.
This is done by
mapping the problem to an instance of 
the {\em min-cut problem}, 
which can be solved in $\bigo{nm}$ time, 
in a recent breakthrough by~\citet{Orlin:2013wu}.
We will present an extension of this transformation in the next section, where
we discuss how to speed-up the algorithm.

Thus, Goldberg's algorithm uses binary search over $\alpha$ and finds 
the largest value of $\alpha^*$ for which
the maximizer set $W^*$ is non empty.
Each iteration of the binary search involves a call to a
min-cut instance for the current value of $\alpha$.

Our algorithm for finding the locally-dense decomposition of a graph 
builds on Goldberg's algorithm~\citep{Goldberg:1984up}.
We show that Goldberg's construction has the following,
rather remarkable, property:
there is a sequence of values 
$\alpha^* =\alpha_1 > \cdots > \alpha_k$, for $k\le n$, 
which gives all the distinct values of the function $\compact{}$.
Furthermore,  
the corresponding set of subgraphs 
$\{\compactgraph{\alpha_1},\ldots,\compactgraph{\alpha_k}\}$ 
is exactly the set of all locally-dense subgraphs of $G$, 
and thus the solution to our decomposition problem.

Therefore, our algorithm is a simple extension of Goldberg's
algorithm: instead of searching only for the optimal value
$\alpha_1=\alpha^*$, 
we find the whole sequence of $\alpha_i$'s
and the corresponding subgraphs.

Next we prove the claimed
properties and discuss the algorithm in more detail.

We first show that the distinct maximizers of the
function~$\compactgraph{}$ correspond to the set of locally-dense
subgraphs.

\begin{proposition}
\label{prop:compact}
Let $\set{B_i}$ be the set of locally-dense subgraphs. 
Then
\[
	B_i = \compactgraph{\alpha}, \quad \text{for} \quad \density{B_{i + 1}, B_{i}} < \alpha \leq \density{B_i, B_{i - 1}}\quad.
\]
\end{proposition}

\begin{proof}
We first show that $U = \compactgraph{\beta}$ is a locally-dense
subgraph, for any $\beta$.
Note that for any $X \subseteq U$, we must have
$\abs{\marginaledges{X, U \setminus X}} - \beta\abs{X} \geq 0$, otherwise we can delete
$X$ from $U$ and obtain a better solution which violates the optimality of $U = \compactgraph{\beta}$.
This implies that $\density{X, U \setminus X} = \marginaledges{X, U \setminus X} / \abs{X} \geq \beta$.
Similarly, for any $Y$ such that $Y \cap U = \emptyset$, we have 
$\abs{\marginaledges{Y, U}} - \beta \abs{Y} < 0$ or, equivalently, 
$\density{Y, U} < \beta$.
Thus, $U$ is locally-dense.

Fix $i$ and select $\alpha$ such that $\density{B_{i + 1}, B_{i}} < \alpha \leq \density{B_i, B_{i - 1}}$.
Let $B_j = \compactgraph{\alpha}$. If $j > i$, then, due to Corollary~\ref{corollary:chain}, $\density{B_j, B_{j - 1}} \leq \density{B_{i + 1}, B_{i}} < \alpha$
which we can rephrase as 
\[
c = \abs{\marginaledges{B_j \setminus B_{j - 1}, B_{j - 1}}} - \alpha|B_j \setminus B_{j - 1}| < 0\quad.
\]
If we delete $B_j \setminus B_{j - 1}$ from $U$, then we improve the quality exactly by $-c$, that is,
we obtain a better solution which violates the optimality of $U$.  If $j < i$,
then Corollary~\ref{corollary:chain} implies that $\density{B_{j + 1}, B_{j}} \geq \alpha$, so 
we can add $B_{j + 1} \setminus B_j$ to obtain a better solution.
It follows that $B_i = \compactgraph{\alpha}$.
\end{proof}

\begin{algorithm}[t]
\caption{\label{algorithm:exact}$\decompose(G,X,Y)$}
\Input{Graph $G = (V, E)$ \\ locally-dense subgraphs $X$ and $Y$ with $X \subsetneq Y$}

$\alpha \define \density{Y, X} + n^{-2}$\;

$Z \define \compactgraph{\alpha}$\;

\If {$Z \neq X$} {
	\Out $Z$\;
	$\decompose(G,X, Z)$\;
	$\decompose(G,Z, Y)$\;
}
\end{algorithm}

Next we need to show that it is possible to search efficiently for the
sequence of $\alpha$'s that give the set of  locally-dense
subgraphs.  
To that end we will show that if we have obtained two subgraphs 
$B_x \subsetneq B_y$ of the decomposition
(corresponding to values  $\alpha_x > \alpha_y$), 
it is possible to pick a new value $\alpha$ so that computing 
$\compactgraph{\alpha}$ allows us to make progress in
the search process: 
we either find a new locally-dense subgraph 
$B_x \subsetneq B_z \subsetneq B_y$ or we establish that no such
subgraph exists between $B_x$ and $B_y$, in other words, 
$B_x$ and $B_y$ are consecutive subgraphs in our decomposition.

\begin{proposition}
\label{prop:rich}
Let $\set{B_i}$ be the set of locally-dense subgraphs. 
Let $B_x \subsetneq B_y$ be two subgraphs.
Set $\alpha = \density{B_y, B_x} + n^{-2}$ and let 
$B_z = \compactgraph{\alpha}$.
If $x + 1 < y$, then $x < z < y$.
If $x + 1 = y$, then $z = x$.
\end{proposition}

\begin{lemma}
\label{lem:chainmonotone}
$\density{B_k, B_i} \geq \density{B_\ell, B_j}$,
for $i \leq j < k \leq \ell$. The equality holds if and only if $i = k$ and $j = \ell$.
\end{lemma}

\begin{proof}
Corollary~\ref{corollary:chain} states that $\density{B_o, B_{o - 1}}$ is monotonically strictly decreasing as a function of $o$.
Lemma~\ref{lem:add}, applied recusively, states that
\[
	\density{B_k, B_i} \geq \density{B_k, B_j} \geq \density{B_\ell, B_j}\quad.
\]
The inequality is strict if and only if $i \neq k$ or $j \neq \ell$. 
\end{proof}

\begin{proof}[Proof of Proposition~\ref{prop:rich}]
Lemma~\ref{lem:chainmonotone} states that $\density{B_y, B_{y - 1}} \leq \density{B_y, B_x} < \alpha$.
Proposition~\ref{prop:compact} now implies that $z < y$.

Assume that $x + 1 < y$. Lemma~\ref{lem:chainmonotone} implies that
$\density{B_y, B_{x}} < \density{B_{x + 1}, B_{x}}$. Write
\begin{align*}
	a & = \abs{\marginaledges{B_y \setminus B_x, B_x}}, & b & = \abs{B_y} - \abs{B_x}, \\
	c & = \abs{\marginaledges{B_{x + 1} \setminus B_x, B_x}}, \quad\text{ and } & d & = \abs{B_{x + 1}} - \abs{B_x}\quad. \\
\end{align*}
Let us now bound the difference between the densities as
\[
	\density{B_{x + 1}, B_{x}} - \density{B_y, B_{x}}
		= \frac{a}{b} - \frac{c}{d}  
		=  \frac{ad - bc}{bd} 
		\geq  \frac{1}{bd}  
		>  \frac{1}{n^2} = \alpha - \density{B_y, B_{x}} \quad.
\]
This implies that $\alpha \leq \density{B_{x + 1}, B_{x}}$.
Proposition~\ref{prop:compact} now implies that $z \geq x + 1 > x$.

Assume that $x + 1 = y$. Lemma~\ref{lem:chainmonotone} implies that $\density{B_y, B_{y - 1}} < \density{B_{x}, B_{x - 1}}$, and
the same argument as above shows that $\alpha \leq \density{B_{}, B_{x - 1}}$ and, consequently, $z \geq x$. This guarantees that $x = z$.
\end{proof}

The exact decomposition algorithm uses Proposition~\ref{prop:rich} to
guide the search process.
Starting by the two extreme subgraphs of the decomposition, 
$\emptyset$ and $V$, the algorithm maintains a sequence of
locally-dense subgraphs.
Recursively, for any two currently-adjacent subgraphs in the sequence, 
we use Proposition~\ref{prop:rich} to check whether the two subgraphs
are consecutive or not in the decomposition. 
If they are consecutive, the recurrence at that branch of the search is
terminated.
If they are not, a new subgraph between the two is discovered and it
is added in the decomposition. 
The algorithm is named \decompose\  and it is illustrated as
Algorithm~\ref{algorithm:exact}.

With the next propositions we prove the correctness of the algorithm
and we bound its running time.

\begin{proposition}
The algorithm \decompose\ initiated with input 
$(G, \emptyset, V)$
visits all non-trivial locally-dense subgraphs of $G$.
\end{proposition}

\begin{proof}
Let $\set{B_i}$ be the set of locally-dense subgraphs.
We will prove the proposition by showing that for $i < j$, the algorithm
$\decompose(G, B_i, B_j)$ visits all monotonic subgraphs that are between $B_i$ and $B_j$.
We will prove this by induction over $j - i$.  The first step $j = i + 1$ is trivial.
Assume that $j > i + 1$. Then Proposition~\ref{prop:rich} implies that
$B_k = \compactgraph{\alpha}$, where $i < k < j$.
The inductive assumption now guarantees that $\decompose(G, B_i, B_k)$
and $\decompose(G, B_k, B_j)$ will visit all monotonic subgraphs
between $B_i$ and $B_j$.  
\end{proof}

\begin{proposition}
The worst-case running time of algorithm \decompose\ is 
$\bigo{n^2m}$.
\end{proposition}

\begin{proof}
We will show that the algorithm \decompose, initiated with input
$(G, \emptyset, V)$ makes $2k - 3$ calls to the function 
$\compactgraph{}$,  
where $k$ is the number of locally-dense subgraphs.

Let $k_i$ be the number of calls of $\compactgraph{}$ when the
input parameter $Y = B_i$.  
Out of these $k_i$ calls one call will result in $\compactgraph{\alpha} = X$.
There are $k - 1$ such calls, since $Y = \emptyset$ is never tested.
Each of the remaining calls will discover a new locally-dense subgraph.
Since there are $k - 2$ new subgraphs to discover, it follows that 
$2k-3$ calls to $\compactgraph{}$ are needed.

Since a call to  $\compactgraph{}$ corresponds to a min-cut
computation, which has running time 
$\bigo{nm}$~\citep{Orlin:2013wu}, and since $k\in\bigo{n}$,
the claimed running-time bound follows.
\end{proof}

\subsection{Speeding up the exact algorithm}

Our next step is to speed-up \decompose. This speed-up does not 
improve the theoretical bound for the computational time but, in
practice, it improves the performance of the algorithm dramatically. 

The speed-up is based on the following observation. We know from
Proposition~\ref{prop:rich} that $\decompose(G, X, Y)$ visits only
subgraphs $Z$ with the property $X \subseteq Z \subseteq Y$. This gives us
immediately the first speed-up: we can safely ignore any vertex outside $Y$, that
is, $\decompose(G(Y), X, Y)$ will yield the same output.

Our second observation is that any subgraph $Z$ visited by $\decompose(G, X,
Y)$ must contain vertices $X$. However, we cannot simply delete them because we
need to take into account the edges between $X$ and $Z$. To address this
let us consider the following maximizer
\[
	\compactgraph{\alpha; X} = \arg \max_{X \subseteq W\subseteq V} \set{ \abs{E(W)} - \alpha \abs{W}}\quad.
\]
We can replace the original $\compactgraph{\alpha}$ in
Algorithm~\ref{algorithm:exact} with $\compactgraph{\alpha; X}$. To compute
$\compactgraph{\alpha; X}$ we will use a straightforward extension of the
Goldberg's algorithm~\citep{Goldberg:1984up} and transform this problem into a
problem of finding a minimum cut.

In order to do this, given a graph $G = (V, E)$, let us define a weighted graph $H$
that consists of vertices $V \setminus X$ and edges $E(V \setminus X)$ with weights of 1.
Add two auxiliary vertices $s$ and $t$ into $H$ and connect these vertices to every vertex in
$V \setminus X$. Given a vertex $y \in V \setminus X$,
assign a weight of $2\alpha$ to the edge $(y, t)$ and a weight of
\[
	w(y) = \dg{y; V \setminus X} + 2\dg{y; X}
\]
to the edge $(s, y)$, where $\dg{y; U}$ stands for the number of neighbors of $y$ in $U$.  We claim that solving a minimum cut such that $s$ and $t$ are in different cuts
will solve $\compactgraph{\alpha; X}$. This cut can be obtained by constructing a maximum flow from $s$ to $t$.

To prove this claim let $C \subsetneq V(H)$ be a subset of vertices containing
$s$ and not containing $t$. Let $Z = C \setminus \set{s}$ and also let $W = V \setminus (Z \cup X)$.
There are three types of cross-edges from $C$ to $V(H) \setminus C$:
\emph{(i)} edges from $x \in Z$ to $t$,
\emph{(ii)} edges from $s$ to $x \in W$, and
\emph{(iii)} edges from $x \in Z$  to $y \in W$. The total cost of $C$ is then
\[
	2\abs{Z}\alpha + \sum_{y \in W} w(y) + \abs{\crossedges{Z, W}}\quad.
\]
We claim that the last two terms of the cost are equal to $2\abs{E} - 2\abs{E(X \cup Z)}$.
To see this, consider an edge $e = (x, y)$ in $E \setminus E(X \cup Z)$. This implies that at least
one of the end points, assume it is $y$, has to be in $W$. There are three different cases for $x$:
\emph{(i)} if $x\in W$, then $e$ contributes 2 to the cost: 1 to $w(x)$ and 1 to $w(y)$,
\emph{(ii)} if $x \in X$, then $e$ contributes $2$ to $w(y)$, and
\emph{(iii)} if $x \in Z$, then $e$ contributes $1$ to $w(y)$ and $1$ to the third term. Thus, we
can write the cut as
\[
	2\abs{Z}\alpha + 2\abs{E} - 2\abs{E(X \cup Z)}  = 2\abs{E} - 2\abs{X}\alpha - 2(\abs{E(X \cup Z)} - \alpha\abs{Z \cup X})\quad.
\]
The first two terms in the right-hand side are constant which implies that
that finding the minimum cut is equivalent of maximizing $\abs{E(X \cup Z)} - \alpha\abs{Z \cup X}$.
Consequently, if $Z^*$ is the min-cut solution, then $\compactgraph{\alpha} = X \cup Z^*$.

Note that the graph $H$ does not have vertices included in $X$.
By combining both speed-ups we are able to reduce the running time of $\decompose(X, Y)$ by considering
only the vertices that are in $Y \setminus X$.

\subsection{Linear approximation algorithm}

As we saw in the last section, the exact algorithm can be
significantly accelerated, and indeed, our experimental evaluation
shows that it is possible to run the exact algorithm for a graph of
millions of vertices and edges within 2 minutes. 
Nevertheless, the worst-case complexity of the algorithm is cubic, and
thus, it is not truly scalable for massive graphs.

Here we present a more lightweight algorithm for performing a
locally-dense decomposition of a graph. 
The algorithm runs in linear time and offers a factor-$2$
approximation guarantee.
As the exact algorithm builds on Goldberg's algorithm for the
densest-subgraph problem, 
the linear-time algorithm builds on Charikar's approximation algorithm
for the same problem~\citep{Charikar:2000tg}.
As already explained in Section~\ref{sec:prel}, 
Charikar's approximation algorithm iteratively removes the vertex with
the lowest degree, until left with an empty graph, and returns the
densest graph among all subgraphs considered during this process.

Our extension to this algorithm, called \peel, is illustrated in
Algorithm~\ref{algorithm:peel}, and it operates in two phases. 
The first phase is identical to the one in Charikar's algorithm:
all vertices of the graph are iteratively removed, 
in increasing order of their degree in the current graph.
In the second phase, the algorithm proceeds to discover approximate
locally-dense subgraphs, in an iterative manner, from $B_1$ to $B_k$.
The first subgraph $B_1$ is the approximate densest subgraph, the same
one returned by Charikar's algorithm.
In the $j$-th step of the iteration, having discover subgraphs 
$B_1,\ldots, B_{j-1}$
the algorithm selects the subgraph $B_j$ that maximizes the density
$\density{B_j,B_{j-1}}$.
To select $B_j$ the algorithm considers subsets of vertices only in
the degree-based order that was produced in the first phase. 

\begin{algorithm}[t]
\caption{\label{algorithm:peel}$\peel(G)$}
\Input{Graph $G = (V, E)$}
\Output{Collection $\col{C}$ of approximate locally-dense subgraphs}
	\For {$i = n, \ldots, 1$} {
		$w_i \define$  the vertex with the smallest degree\;
		delete $w_i$ from $V$\;
	}
	$\col{C} \define \set{\emptyset}$\;		
	$j \define 0$\;
	\While {$j < n$} {
		$i \define \arg \max_{i > j}  \density{\enset{w_1}{w_i}, \enset{w_1}{w_j}}$\;
		add $\enset{w_1}{w_i}$ to $\col{C}$\;
		$j \define i$\;
	}
	\Return $\col{C}$\;
\end{algorithm}

Discovering $\col{C}$ from the ordered vertices takes $\bigo{n^2}$ time, if
done naively. However, it is possible to implement this step in $\bigo{n}$
time. In order to do this, sort vertices in the reverse visit order,
and define $\din{v}$ to be the number of edges of
$v$ from the earlier neighbors. Then, we can we express the density as an average,
\[
	\density{\enset{w_1}{w_i}, \enset{w_1}{w_j}} = \frac{1}{i - j}\sum_{k = j + 1}^i \din{v_k}\quad.
\]
Consequently, we can see that recovering $\col{C}$ is an instance of the following problem,
\begin{problem}
Given a sequence $y_1, \ldots, y_n$, compute the maximal interval
\[
	m(j) = \arg \max_{j \leq i \leq n} \frac{1}{i - j + 1}\sum_{k = j}^i y_k,
	\quad\text{for every}\quad 1 \leq j \leq n\quad.
\]
\end{problem}

Luckily, \citet{DBLP:journals/is/CaldersDGG14} demonstrated that we can use the classic PAVA
algorithm by~\citet{PAV} to solve this problem for \emph{every} value of $j$ in total $\bigo{n}$ time. 

To quantify the approximation guarantee of \peel, 
note that the sequence of approximate locally-dense subgraphs produced
by the algorithm are not necessarily aligned with the locally-dense
subgraphs of the optimal decomposition. 
In other words, to assess the quality of the density of an approximate
locally-dense subgraph $B_j$ produced by \peel, 
there is no direct counterpart in the optimal decomposition to compare.
To overcome this difficulty we develop a scheme of ``vertex-wise''
comparison, where for any $1\le i\le n$, 
the density of the smallest approximate locally-dense subgraph of
size at least $i$  
is compared with the density of the smallest optimal locally-dense
subgraph of size at least $i$.  
This is defined below via the concept of {\em profile}.

\begin{definition}
\label{definition:profile}
Let $\mathcal{B} = (\emptyset = B_0 \subsetneq B_1 \subsetneq \cdots
\subsetneq B_k = V)$ be a nested chain of subgraphs, the first
subgraph being the empty graph and the
last subgraph being the full graph. 
For an integer $i$, $1 \leq i \leq n$ define
\[
	j = \min \set{x \mid \abs{B_x} \geq i}
\]
to be the index of the smallest subgraph in $\mathcal{B}$ whose size
is at least $i$. We define a {\em profile function}
$\funcdef{\prof{}}{\enset{1}{n}}{\real}$ to be 
\[
	\prof{i;\,\mathcal{B}} = \density{B_j, B_{j - 1}}\quad.
\]
\end{definition}

Our approximation guarantee is now expressed as a guarantee of the
profile function of the approximate decomposition with respect to the
optimal decomposition.

\begin{proposition}
\label{prop:peel-approximation}
Let $\col{B} = \set{B_i}$ be the set of locally-dense subgraphs.
Let $\col{C} = \set{C_i}$ be the subgraphs obtained by \peel.
Then
\[
	\prof{i;\,\col{C}} \geq \prof{i;\,\col{B}} / 2\quad.
\]
\end{proposition}

First, we need the following lemma.

\begin{lemma}
\label{lem:degree}
$\density{v, B_i \setminus \set{v}} \geq \density{B_i, B_{i - 1}}$,
for $v \in B_i \setminus B_{i - 1}$,
\end{lemma}

\begin{proof}
Assume otherwise. Lemma~\ref{lem:add} now states that
$\density{B_i \setminus \set{v}, B_{i - 1}} < \density{B_i, B_{i - 1}}$,
which violates the optimality of $B_i$ as indicated by
Proposition~\ref{prop:maximal}.
\end{proof}

\begin{proof}[Proof of Proposition~\ref{prop:peel-approximation}]
Sort the set of vertices $V$ according to the reverse visiting order
of \peel and let $\din{v}$ be the number of edges of $v$
from earlier neighbors.

Fix $k$ to be an integer, $1 \leq k \leq n$ and let $B_i$ be the smallest subgraph
such that $\abs{B_i} \geq k$.  Let $v_j$ be the last vertex occurring in
$B_i$.  We must have $\din{v_j} \geq \density{v_j, B_i \setminus \set{v_j}}$, and, due to Lemma~\ref{lem:degree}, $\density{v_j,
B_i \setminus \set{v_j}} \geq \density{B_i , B_{i - 1}}$. In summary, we have
\[
	\prof{k;\,\mathcal{B}} = \density{B_i, B_{i - 1}} \leq \din{v_j}\quad.
\]

Let $C_x$ be the smallest subgraph such that $\abs{C_x} \geq k$.
Let $v_z$ be the vertex with the smallest index that is still in $C_x \setminus C_{x - 1}$
and define $A = \set{v_z, \ldots, v_j}$.
Let $g(v)$ be the degree of $v \in A$ right before $v_j$ is removed during \peel. 
Note that, by definition, $\din{v_j} \leq g(v)$, and that 
\[
	\sum_{v \in A} g(v) = 2\abs{\edges{A}} + \abs{\crossedges{A, C_{x - 1}}} 
	\leq 2\abs{\edges{A}} + 2\abs{\crossedges{A, C_{x - 1}}} = 2\sum_{v \in A} \din{v}
	\quad.
\]
This leads to
\[
\prof{k;\,\mathcal{C}}  =  \density{C_x, C_{x - 1}}  
 \geq  \density{A, C_{x - 1}} = \frac{1}{\abs{A}} \sum_{v \in A} \din{v} 
 \geq  \frac{1}{2\abs{A}} \sum_{g(v) \in A} g(v)  
 \geq  \frac{\din{v_j}}{2},
\]
where the optimality of $C_x$ implies the first inequality.
\end{proof}
We should point out that $\prof{1, \col{B}}$ is equal to the density of the
densest subgraph, while $\prof{1, \col{C}}$ is equal to the density of the
subgraph discovered by the Charikar's algorithm. Consequently,
Proposition~\ref{prop:peel-approximation} provides automatically the
2-approximation guarantee of the Charikar's algorithm.

We should also point out that $\prof{i, \col{C}}$ can be larger than $\prof{i,
\col{B}}$. However, if $j$ is the \emph{first} index, for which $\prof{j, \col{C}} \neq
\prof{j, \col{B}}$, then Proposition~\ref{prop:maximal} guarantees that $\prof{j,
\col{C}} < \prof{j, \col{B}}$.

\section{Locally-dense subgraphs and core decomposition}
\label{sec:core}

Here we study the connection of graph cores, 
obtained with the well-known $k$-core decomposition algorithms, 
with local-density, studied in this paper. 
We are able to show that from the theory point-of-view, 
graph cores are as good approximation to the optimal locally-dense
graph decomposition as the subgraphs obtained by the \peel algorithm.
In particular we show a similar result to
Proposition~\ref{prop:peel-approximation}, 
namely, a factor-$2$ approximation on the profile function of the core
decomposition.

However, as we will see in our empirical evaluation, 
the behavior of the two algorithms, 
\peel and $k$-core decomposition are different in practice, 
with \peel giving in general more dense subgraphs and closer to the
ones given by exact locally-dense decomposition.

Before stating and proving the result regarding $k$-cores, 
recall that a set of vertices $X\subseteq V$ is a $k$-core if 
every vertex in the subgraph induced by $X$ has degree at least $k$,
and $X$ is maximal with respect to this property. 
A linear-time algorithm for obtaining all $k$-cores is illustrated in 
Algorithm~\ref{algorithm:k-core}. 

\begin{algorithm}[t]
\caption{\label{algorithm:k-core}$\core(G)$}
\Input{Graph $G = (V, E)$}
\Output{Collection $\mathcal{C}$ of $k$-cores}
   $\col{C} \define \set{V}$\;
   $k \define \min_w \dg{w}$\;
   \For {$i = n, \ldots, 1$} {
        $w_i \define$  the vertex with the smallest degree\;
		\If {$\dg{w} > k$} {
			add $V$ to $\col{C}$\;
			$k \define \dg{w}$\;
		}
        delete $w_i$ from $V$\;
    }
    \Return $\mathcal{C}$\;
\end{algorithm}

It is a well-known fact that the set of all $k$-cores of a graph forms
a nested chain of subgraphs, in the same way that locally-dense
subgraphs do. 

\begin{proposition}
\label{prop:core-approximation}
Let $\set{C_i}$ be the set of all $k$-cores of a graph $G=(V,E)$.
Then $\set{C_i}$ forms a nested chain,
\[
	\emptyset = C_0\subsetneq C_1 \subsetneq \cdots \subsetneq C_l = V\quad.
\]
\end{proposition}

Similar to 
Proposition~\ref{prop:peel-approximation}, 
$k$-cores provide a factor-$2$ approximation 
with respect to the locally-dense subgraphs.
The proof is in fact quite similar to that of 
Proposition~\ref{prop:peel-approximation}.

\begin{proposition}
Let $\col{B} = \set{B_i}$ be the set of locally-dense subgraphs.
Let $\col{C} = \set{C_i}$ be the set of $k$-cores.
Then
\[
	\prof{i;\,\col{C}} \geq \prof{i;\,\col{B}} / 2\quad.
\]
\end{proposition}

\begin{proof}
Sort $V$ according to the reverse visiting order of \core
and let $\din{v}$ be the number of edges of $v$
from earlier neighbors.

Fix $k$ to be an integer, $1 \leq k \leq n$ and let $B_i$ be the smallest subgraph
such that $\abs{B_i} \geq k$.  Let $v_j$ be the last vertex occurring in
$B_i$.  We must have $\din{v_j} \geq \density{v_j, B_i \setminus \set{v_j}}$, and, due to Lemma~\ref{lem:degree}, $\density{v_j,
B_i \setminus \set{v_j}} \geq \density{B_i , B_{i - 1}}$. In summary, we have
\[
	\prof{k;\,\mathcal{B}} = \density{B_i, B_{i - 1}} \leq \din{v_j}\quad.
\]

Let $C_x$ be the smallest core such that $\abs{C_x} \geq k$, and write $A = C_x \setminus C_{x - 1}$.
Let $v_s$ be the vertex with the smallest index that is still in $A$, and
let $v_l$ be the vertex with the largest index that is still in $A$,
that is, $\set{v_s, \ldots, v_l} = A$.

If $j > l$, then $\din{v_j} < \din{v_l}$, otherwise $C_x$ is not a core.
If $j < l$, then $\din{v_j} \leq \din{v_l}$, otherwise $v\j \notin C_x$,
and since $j \geq k$, then $C_x$
is not the smallest core  with at least $k$ vertices, which is a contradiction.
Hence, $\din{v_j} \leq \din{v_l}$.

Let $g(v)$ be the degree of $v \in A$ right before $v_l$ is removed during \core.
We now have
\begin{eqnarray*}
\prof{k;\,\mathcal{C}} & = & \density{C_x, C_{x - 1}} \\
&  = & \frac{1}{\abs{A}} \sum_{v \in A} \din{v} \\
& \geq & \frac{1}{2\abs{A}} \sum_{v \in A} g(v) \\ 
& \geq & \frac{\din{v_l}}{2} \\ 
& \geq & \frac{\din{v_j}}{2},
\end{eqnarray*}
which proves the proposition.
\end{proof}

\section{Segmentation problem: constraining the number of subgraphs}
\label{sec:segmentation}

It is possible that the decomposition yields a significant amount of subgraphs.
In such a case it may be useful to constraint the number of the subgraphs.  In
order to do so we need to define an optimization criterion, which will be our
first step.  We then demonstrate how to solve the problem exactly, and how to
estimate the solution efficiently. 

\subsection{Problem definition}

Our goal is to discover $k$ nested subgraphs that minimize a certain cost. We
base the cost on the degree of a node, relative to the subgraph.  A natural
approach here is to model the degree, that is, our goal is to maximize the
log-likelihood $\sum_v \log p(\dg{v; C_i}; \lambda_i)$, where $C_i$ is the
smallest subgraph containing $v$ and $\lambda_i$ is a parameter of the
distribution.  Unfortunately, this is problematic due to the following reason:
an edge $(x, y)$, where $x, y \in C_i \setminus C_{i - 1}$ increases the
degrees of both $x$ and $y$, whereas an edge $(x, y)$, with $x \in C_i$ and $y
\in C_{i - 1}$ increases the degrees only for $x$ and \emph{not} for $y$. The
distribution we will consider favors small degrees, so this leads to a scenario
where the cost function implicitly favors having a lot of cross-edges.
To rectify this problem we introduce the notion of \emph{adjusted degree}, where we
count each cross-edge twice.

\begin{definition}
Assume a sequence of nested subgraphs
$\col{C} = \pr{\emptyset = C_0 \subsetneq C_1 \subsetneq \cdots \subsetneq C_k = V}$.
Let $v$ be a vertex and let $C_i$ be the smallest set containing $v$.
Define the \emph{adjusted degree} as
\[
	\adg{v; \col{C}} =
	\abs{\set{(v, u) \mid u \in C_i \setminus C_{i - 1}}} + 
	2\abs{\set{(v, u) \mid u \in C_{i - 1}}}\quad.
\]
\end{definition}

To reduce the clutter, we typically omit $\col{C}$ from the notation and write $\adg{v}$.

Next we give a formal definition of the problem.

\begin{definition}
Assume that we are given a distribution $p(\cdot ; r)$
for the adjusted degree. This distribution has one parameter $r$;
small values indicate the likelihood of high degrees.
Given a graph $G$ and an integer $k$, find a $k$-\emph{segmentation}, a sequence of nested subgraphs
$\col{C} = \pr{\emptyset = C_0 \subsetneq C_1 \subsetneq \cdots \subsetneq C_k = V}$
and parameters $\lambda_1 \leq \cdots \leq \lambda_k$,
minimizing the negative log-likelihood
\[
	\cost{\col{C}} = -\sum_{v \in V} \log p(\adg{v}; \lambda_i),
\]
where $i$ is the index of the smallest $C_i$ containing $v$.
\end{definition}
The reason why we write this problem as a minimization problem is because
typically the log-likelihood is negative, and in order to establish approximation guarantees
we need to have the cost function to be positive.

We are specifically interested in geometric and exponential distributions. Both
distributions can be written as $p(x; \lambda) = \exp(-\lambda x -
Z(\lambda))$, where $Z(\lambda)$ is the normalization constant\footnote{The
geometric distribution is defined over the integers whereas the exponential
distribution is defined over the real domain. This results in different
normalization constants.}.  Moreover, smaller values of $\lambda$ will result
in a distribution favoring larger degrees, that is, inner subgraphs should be denser.

\subsection{Exact algorithm}

In this section we demonstrate how to find an optimal segmentation using
locally-dense subgraphs. First we prove the key proposition that states that it
is enough to use locally-dense subgraphs when looking for the optimal
segmentation.

\begin{proposition}
\label{prop:segment}
Assume that $p$ is either exponential or geometric distribution.
Then there is an optimal segmentation
$\col{C} = \pr{\emptyset = C_0 \subsetneq C_1 \subsetneq \cdots \subsetneq C_k = V}$
such that each $C_i$ is locally-dense.
\end{proposition}

To prove the proposition, we need the following technical lemma.

\begin{lemma}
\label{lem:violate}
Let $C_1, \ldots, C_k$ be the optimal solution, and assume
some of the subgraphs are not locally-dense.
Then there is $C_i$ that is not locally-dense along with the violating sets $X$ and $Y$ such that $Y
\subseteq C_{i + 1}$ and $X \cap C_{i - 1} = \emptyset$.
\end{lemma}

\begin{proof}
Let $C_i$ be a set that is not locally-dense, and let $X$ and $Y$ be the violating sets. 
Next we argue that we can safely assume that $Y \subseteq C_{i + 1}$ and
$X \cap C_{i - 1} = \emptyset$. We will split the argument in two cases:
Case ($i$): $Y \nsubseteq C_{i + 1}$ and
Case ($ii$): $Y \subseteq C_{i + 1}$.

Assume Case ($i$). If $\density{C_{i + 1} \setminus C_i, C_i} \geq \density{X, C_i
\setminus X}$, then redefine $Y$ as $C_{i + 1} \setminus C_i$. In such case,
$X$ and $Y$ are still violating the local density but now we can use Case
($ii$). Assume that $\density{C_{i + 1} \setminus C_i, C_i} < \density{X, C_i
\setminus X}$.
Define $Y_1 = Y \cap C_{i + 1}$ and $Y_2 = Y \setminus Y_1$.
Note that $Y_2 \neq \emptyset$. Assume that $\density{Y_2, C_i \cup Y_1} \geq \density{Y, C_i}$.
Then
\[
	\density{Y_2, C_{i + 1}} \geq \density{Y_2, C_{i} \cup Y_1}  \geq \density{Y, C_i} \geq \density{X, C_i
	\setminus X} > \density{C_{i + 1} \setminus C_i, C_i}\quad.
\]
Redefine $Y$ as $Y_2$, $X$ as $C_{i + 1} \setminus C_i$, and increase $i$ by 1.
The previous arguments show that new $Y$ and $X$ violate the local density of $C_{i + 1}$,
so we repeat our argument with either Case ($i$) or Case ($ii$).

Assume now that $\density{Y_2, C_i \cup Y_1} < \density{Y, C_i}$.
This forces $Y_1 \neq \emptyset$.
Since $\density{Y, C_i}$ is a weighted average of $\density{Y_2, C_i \cup Y_1}$ and $\density{Y_1, C_i}$,
we have
$\density{Y, C_i} \leq \density{Y_1, C_i}$. Redefine $Y$ as $Y_1$, and apply Case ($ii$).

Assume Case ($ii$).
Write $X_1 = X \cap C_{i - 1}$ and $X_2 = X \setminus X_1$.
If $X_1 = \emptyset$, then we are done; assume otherwise.
If $\density{C_{i} \setminus C_{i - 1}, C_{i - 1}} \leq \density{Y, C_{i}}$, then we can replace $X$
with $C_{i} \setminus C_{i - 1}$ to complete the argument. Assume that $\density{C_{i} \setminus C_{i - 1}, C_{i - 1}} > \density{Y, C_{i}}$.

Assume $X_2 \neq \emptyset$.
If $\density{X_2, C_{i} \setminus X_2} \leq \density{X, C_{i} \setminus X}$, then we can replace $X$ with $X_2$
to complete the argument.
Assume $\density{X_2, C_{i} \setminus X_2} > \density{X, C_{i} \setminus X}$.
Note that $\density{X, C_{i} \setminus X}$ is a weighted average of
$\density{X_2, C_{i} \setminus X_2}$ and $\density{X_1, C_{i} \setminus X}$.
This implies that $\density{X_1, C_{i} \setminus X} < \density{X, C_{i} \setminus X}$.

On the other hand, if $X_2 = \emptyset$, then $X_1 = X$, and $\density{X_1, C_{i} \setminus X} = \density{X, C_{i} \setminus X}$.

Combining everything gives us
\[
	\density{X_1, C_{i - 1} \setminus X_1} \leq \density{X_1, C_{i} \setminus X} \leq \density{X, C_{i} \setminus X} \leq \density{Y, C_{i}} < \density{C_{i} \setminus C_{i - 1}, C_{i - 1}}\quad.
\]
Redefine $X$ as $X_1$, $Y$ as $C_{i} \setminus C_{i - 1}$ and decrease $i$ by one, and repeat Case ($ii$).

Note that we do first at most $k$ repetitions of Case ($i$), and then at most $k$ repetitions of Case ($ii$).
After a finite numer of repetitions we end up with $C_i$ that satisfies the conditions.
This completes the proof.
\end{proof}

\begin{proof}[Proof of Proposition~\ref{prop:segment}]
Both geometric and exponential distributions can be written as $p(x; \lambda) = \exp(-\lambda x - Z)$,
where $Z$ is the normalization constant (depending on $\lambda$).

Write $B_i = C_i \setminus C_{i - 1}$.
We can write the optimization function as
\[
	\cost{\col{C}}
	= \sum_{i = 1}^k \sum_{v \in B_i} Z_i + \lambda_i \adg{v; C_i}
	= \sum_{i = 1}^k \abs{B_i}(Z_i + 2\lambda_i \density{B_i, C_{i - 1}}),
\]
where $Z_i$ is the normalization constant for the parameter $\lambda_i$.

Assume that $C_i$ is not locally-dense, that is, there is $X$ and $Y$ that
violate the local density. Lemma~\ref{lem:violate} states that we can safely assume that
$Y \subseteq C_{i + 1}$ and $X  \cap C_{i - 1} = \emptyset$. 
This allows us to either remove $X$ from $C_i$ or add $Y$ to $C_i$ without changing the other sets.

The cost of the $i$th and the $i + 1$th segment is equal to 
\[
	\abs{B_i}Z_i + \abs{B_{i + 1}}Z_{i + 1} + 2\lambda_i \marginaledges{B_i, C_{i - 1}} + 2\lambda_{i + 1} \marginaledges{B_{i + 1}, C_{i}} \quad . 
\]
Let us define $W = B_i \cup B_{i + 1}$. Due to the equality
\begin{equation}
\label{eq:marginal} 
	\marginaledges{I, A} - \marginaledges{J, A} = \marginaledges{I \setminus J, A \cup J},
	\quad\text{for}\quad J \subset I,
\end{equation}
the cost can be rewritten as
\[
	\abs{B_i}(Z_i - Z_{i + 1}) + \abs{W}Z_{i + 1} + 2\lambda_i \marginaledges{B_i, C_{i - 1}} + 2\lambda_{i + 1} (\marginaledges{W, C_{i - 1}} -  \marginaledges{B_{i}, C_{i - 1}} )
\]
by setting $I = W$, $J = B_i$ and $A = C_{i - 1}$.
We would like to vary $B_i$ while keeping the remaining variables constant; let us define
\[
\begin{split}
	g(U) & = \abs{U}(Z_i - Z_{i + 1}) + \abs{W}Z_{i + 1} + 2\lambda_i \marginaledges{U, C_{i - 1}} + 2\lambda_{i + 1} (\marginaledges{W, C_{i - 1}} -  \marginaledges{U, C_{i - 1}} ) \\
         & = \abs{U}(Z_i - Z_{i + 1}) - 2(\lambda_{i + 1} - \lambda_{i}) \marginaledges{U, C_{i - 1}} + \abs{W}Z_{i + 1} + 2\lambda_{i + 1} \marginaledges{W, C_{i - 1}}\quad.
\end{split}
\]
Note that the last two terms do not depend on $U$.
Due to optimality of $C_i$, we have $g(B_i) \leq g(B_i \cup Y)$,
or
\[
\begin{split}
	0  & \leq g(B_i \cup Y) - g(B_i) \\
	   & = \abs{Y}(Z_i - Z_{i + 1}) - 2(\lambda_{i + 1} - \lambda_{i}) (\marginaledges{B_i \cup Y, C_{i - 1}} - \marginaledges{B_i, C_{i - 1}}) \\
	   &  = \abs{Y}(Z_i - Z_{i + 1}) - 2(\lambda_{i + 1} - \lambda_{i}) \marginaledges{Y, C_{i}}, 
\end{split}
\]
where the last equality is due to Eq.~\ref{eq:marginal}.  We can rewrite the inequality as
\[
	Z_{i} - Z_{i + 1} \geq 2(\lambda_{i + 1} - \lambda_i)\density{Y, C_{i}} \geq 2(\lambda_{i + 1} - \lambda_i)\density{X, C_{i} \setminus X},
\]
where the last inequality follows from the fact that $X$ and $Y$ violate the local density of $C_i$, and
since $\lambda_{i + 1} \geq \lambda_i$.
We can rewrite the left-hand side and the right-hand side as
\[
	0 \leq  \abs{X}(Z_{i + 1} - Z_{i}) - 2(\lambda_{i + 1} - \lambda_{i}) \marginaledges{X, C_{i} \setminus X}
	= g(B_i) - g(B_i \setminus X),
\]
or $g(B_i \setminus X) \leq g(B_i)$.

We have shown that if there is $C_i$ that is not locally-dense, we can delete some vertices from $C_i$ without sacrificing the quality. 
We continue this until all $C_i$ are locally-dense; the process must end because at each step we reduce the size of some $C_i$.
\end{proof}

The proposition gives us means to compute the optimal segmentation.  First we
discover locally-dense decomposition, say, $\mathcal{L}$.  If the number of
subgraphs is less or equal than $k$, we are done.  Otherwise, we group
subgraphs until we reach $k$. The optimal grouping can be done with a dynamic
program. Write $o[i, j]$ to be the cost of partial $j$-segmentation using only
$L_0, \ldots, L_i$. We have the identity
\[
	o[i, j] = \min_{\ell < i} c[\ell, i] + o[\ell, j - 1],
	\quad\text{where}\quad
	c[\ell, i] = -\sum_{v \in L_i \setminus L_\ell} \log p(\adg{v}; \lambda)
\]
and $\lambda$ is the optimal parameter for modeling $L_i \setminus L_\ell$.
This identity allows us to compute $o[n, k]$ recursively with a dynamic program.
Note that the monotonicity of the segmentation---that is, the inner subgraphs
should be more dense---is automatically guaranteed.
We will refer to this algorithm as $\segment(\col{L}, k)$.

Computing $c[\ell, i]$ can be done in constant time. To see this, let $r =
\abs{L_i \setminus L_\ell}$ be the number of nodes in $L_i \setminus L_\ell$.
Let also 
\[
	q = \sum_{v \in \in L_i \setminus L_\ell} \adg{v} = 2 \marginaledges{L_i \setminus L_\ell, L_\ell}
\]
be the sum of all
adjusted degrees in $L_i \setminus L_\ell$. Note $r$ and $q$ can be maintained in constant time.
Then the corresponding costs for the geometric and exponential distributions are
\[
	c_{\mathit{geo}}[i, \ell] = -r \log \frac{r}{r + q} - q \log \frac{q}{r + q} \quad\text{and}\quad
	c_{\mathit{exp}}[i, \ell] = r + r\log\frac{q}{r} \quad.
\]

Let us consider computational complexity. Discovering locally-dense decomposition
can be done in $\bigo{n^2m}$ time, whereas the actual segmentation can be done in $\bigo{\ell^2k} \subseteq \bigo{n^2k}$ time,
where $\ell$ is the number of subgraphs in locally-dense decomposition.
In practice, $\ell \ll n$ so the segmentation step is relatively cheap. However, if $\ell$ is large,
it is possible to achieve $(1 + \epsilon)$ approximation for the segmentation in linear time~\citep{guha:06:estimate,tatti19segmentation}.

\subsection{Approximation algorithm}

As pointed out above, the bottleneck of the exact algorithm is the
locally-dense decomposition step. For large graphs we can significantly speed-up this step by
using the faster algorithm \peel.
The next proposition shows that this yields 2-approximation guarantee,
if we use the geometric distribution.

\begin{proposition}
\label{prop:geometrysegment}
Let $p$ be the geometric distribution.
Let $\mathcal{C}$ be the optimal segmentation, and let
$\mathcal{O} = \segment(\peel(G), k)$ be the optimal segmentation using the sets obtained from \peel.
Then $\cost{\mathcal{O}} \leq 2\cost{\mathcal{C}}$.
\end{proposition}

Before proving the result, we need to introduce some notation.
The geometric distribution can be written as
\[
	-\log p(x; \lambda) = \lambda x + Z(\lambda),
\]
where $Z(\lambda) \geq 0$ is the normalization constant.

To prove the result let us enumerate the vertices, that is, $V = \set{v_i}_{i = 1}^n$,
and assume that this order respects the optimal segmentation $\set{C_i}$, $v_i \in C_j$
implies that $v_{i - 1} \in C_j$.
Let $\set{\lambda_i}$ be the optimal parameters for $\set{C_i}$.
We write $\eta_i$ to be the parameter $\lambda_j$ that is used to model $\adg{v_i; C_j}$, where $C_j$ is the smallest subgraph containing $v_i$.
Write $Z = \sum Z(\eta_i)$ to be the sum of normalization constants. Note that $Z \geq 0$.
Given a sequence $X = x_1, \ldots, x_n$, we define
\[
	f(X) = Z + \sum_{i = 1}^n x_i \eta_i\quad.
\]
Define $A$ with $a_i = \adg{v_i}$.
Note that $f(A) = \cost{\mathcal{C}}$.

Define an order for vertex indices $\pr{o_i}_{i = 1}^n$, vertices with high
degree first, that is, $\dg{v_{o_i}} \geq \dg{v_{o_{i + 1}}}$.
Define a sequence $T$ with $t_i = \dg{v_{o_i}}$.

\begin{lemma}
	$f(T) \leq  \cost{\mathcal{C}}.$
\end{lemma}

\begin{proof}
Define $T'$ as $t'_i = \dg{v_i}$.  We argue first that $f(T') \leq \cost{\mathcal{C}} = f(A)$.
We can rewrite
\[
	f(T') = Z + \sum_{(v_i, v_j) \in E} \eta_i + \eta_j
\quad\text{and}\quad
	f(A) = Z + \sum_{(v_i, v_j) \in E} 2\eta_{\max(i, j)}\quad.
\]
Since $\eta_i \leq \eta_{i + 1}$, we have $f(T') \leq f(A)$.
To prove $f(T) \leq f(T')$, note that
\[
	x \alpha + y \beta \leq y \alpha + x \beta, \quad\text{for}\quad \alpha \leq \beta, \ x \geq y\quad.
\]
That is, let $\pr{q_i}$ be \emph{any} vertex order, and let $X$ be the degree sequence
$x_i = \dg{v_{q_i}}$. Then sorting the vertices with bubble sort from $\pr{q_i}$ to $\pr{o_i}$
will not increase the sum in $f$ at any step. Consequently, $f(T) \leq f(X)$. Since this
holds for any order, $f(T) \leq f(T')$,
which proves the lemma.
\end{proof}

Let $\pr{g_i}_{i = 1}^n$ be the reverse order of indices in which \peel removes the vertices, and
let $s_i$ be the degree of $v_{g_i}$ during its removal.

\begin{lemma}
$s_i \leq t_i$.
\end{lemma}

\begin{proof}
Consider two \emph{sets} $P = \set{v_{g_1}, \ldots, v_{g_{i - 1}}}$
and $Q = \set{v_{o_1}, \ldots, v_{o_{i - 1}}}$.
Assume that $P \neq Q$ when treated as sets, that is,
there are indices $j$ and $\ell$ with $j < i \leq \ell$
such that $g_j = o_\ell$. Let $h$ be the degree of $v_{g_j}$
when deleting $v_{g_i}$. Since \peel deletes the vertex with the smallest degree,
$s_i \leq h$.  Consequently, $s_i \leq h \leq \dg{v_{g_j}} = t_\ell \leq t_i$.

Assume the opposite case: $P = Q$.
Due to pigeonhole principle,
there is $\ell \geq i$ such that $o_\ell = g_i$.
Thus, $s_i \leq \dg{v_{g_i}} = t_\ell \leq t_i$. 
\end{proof}

\begin{proof}[Proof of Proposition~\ref{prop:geometrysegment}]
Define $B$ as $b_i = \adg{v_{g_i}}$. Note that $\adg{v_{g_i}} \leq 2s_i$. Thus,
\[
	2f(T) = 2Z + 2\sum_{v_i \in V} \eta_i t_i
	\geq 2Z + 2\sum_{v_i \in V} \eta_i s_i
	\geq 2Z + \sum_{v_i \in V} \eta_i b_i \geq Z + \sum_{v_i \in V} \eta_i b_i = f(B)\quad.
\]
Consider a segmentation $\mathcal{G}$ respecting the order $g_i$ and having the
same sizes as $\mathcal{C}$, $\abs{C_i} = \abs{G_i}$. The value $f(B)$
corresponds to the log-likelihood of $\mathcal{G}$ and the parameters
$\lambda_1, \ldots, \lambda_k$, and $\cost{\mathcal{G}}$ corresponds to the
log-likelihood of $\mathcal{G}$ and the optimized parameters. Thus,
$\cost{\mathcal{G}} \leq f(B)$.

We have shown that there is a segmentation respecting the order chosen by \peel
that is at most $2\cost{\mathcal{C}}$. Thus, the optimal segmentation
respecting the order is also at most $2\cost{\mathcal{C}}$.
The argument in the proof of Proposition~\ref{prop:segment} can be now used to show
that we can safely assume that the segmentation uses sets returned by \peel.
\end{proof}

We can show a similar result for the exponential distribution
as long as the original graph does not have any singletons.

\begin{proposition}
\label{prop:expsegment}
Let $p$ be the exponential distribution. Assume that $G$ has no singletons.
Let $\mathcal{C}$ be the optimal segmentation, and let
$\mathcal{O} = \segment(\peel(G), k)$ be the optimal segmentation using the sets obtained from \peel.
Then $\cost{\mathcal{O}} \leq 2\cost{\mathcal{C}}$.
\end{proposition}

\begin{proof}
Similarly to the geometric distribution, exponential distribution can be written as
\[
	-\log p(x; \lambda) = \lambda x + Z(\lambda)\quad.
\]
Let $Z$ be as defined in proof of Proposition~\ref{prop:geometrysegment}, that is, it is total sum
of the normalization constants. To prove the result we only need to show that $Z \geq 0$,
and we can use the proof of Proposition~\ref{prop:geometrysegment}. Note that $Z(\lambda) = -\log \lambda$,
and the optimal $\lambda$ for a segment $C_i$ is $1 / [2\density{C_i \setminus C_{i - 1}, C_{i - 1}}]$.
This leads to
\[
	Z = \sum_{i = 1}^k \abs{C_i} \log 2\density{C_i \setminus C_{i - 1}, C_{i - 1}}\quad.
\]
To prove the result we will show that $\density{C_i \setminus C_{i - 1}, C_{i - 1}} \geq 1 / 2$.
It is enough to prove the case $i = k$ as due to Proposition~\ref{prop:segment}
the densities are monotonic.

Let $X$ be any subset of vertices. As there are no singletons, $\dg{v} \geq 1$. This leads to
\[
	\density{X, V \setminus X} \geq \frac{1}{2 \abs{X}} \sum_{v \in X} \dg{v} \geq \frac{1}{2}\quad.
\]
Set $X = C_k \setminus C_{k - 1}$ to complete the proof.
\end{proof}

We should point out that these results also work if the graph has weights on
the edges. However, in such a case, Proposition~\ref{prop:expsegment} requires
weights to be larger than or equal to 1.

\section{Related work}
\label{sec:related}

This paper is an extension of previouly published
work~\citep{tatti:2015:density}, and in this extension we introduce the
segmentation problem, where we constrain the number of subgraphs.
\citet{danisch2017large} introduced an alternative iterative technique for
computing locally-dense decomposition that scales well in practice.

Our paper is related to previous work on discovering dense subgraphs, 
clique-like structures, and hierarchical communities. 
We review some representative work on these topics. 

\spara{Clique relaxations.}
The densest possible subgraph is a clique. 
Unfortunately finding large cliques is computationally
intractable~\citep{DBLP:conf/focs/Hastad96}. 
Additionally, the notion of clique does not provide a robust
definition for practical situations, as a few absent edges may
completely destroy the clique. 
To address these issues, researchers have come up with relaxed
clique definitions.  
A relaxation, $k$-plex was suggested by~\citet{seidman10kflex}. 
In a $k$-plex a vertex can have at most $k - 1$ absent edges.  
Unfortunately, discovering maximal $k$-plexes is also an \np-hard
problem~\citep{balasundaram:2011:kflex}. 
An alternative relaxation for a clique is the one of an $n$-clique, 
a maximal subgraph where each vertex is connected
to every vertex with a path, possibly outside of the subgraph, of at most
$n$-length~\citep{Bron:1973:AFC:362342.362367}.
So, according to this definition a clique is an $1$-clique.  
As maximal $n$-cliques may produce sparse graphs,
the concept of $n$-clans was also proposed by limiting the diameter of
the subgraph to be at most $n$~\citep{mokken:79:clans}. Since $1$-clan corresponds
to a maximal clique, discovering $n$-clans is a computationally intractable problem.

\spara{Quasi-cliques.}
For the definition of graph density we have chosen to
work with $\density{X}$, the average degree of the subgraph induced by
$X$. While this is a popular density definition, there are other alternatives. 
One such alternative would be to divide the number of edges present in the
subgraph with the total number of possible edges, that is,  
divide by ${n \choose 2}$.
This would give us a normalized density score that is 
between $0$ and $1$. 
Subgraphs that maximize this density definition are called 
{\em quasi-cliques}, and algorithms for enumerating all 
quasi-cliques, which can be exponentially many,  
have been proposed by~\citet{abello02clique} and~\citet{Uno:2010:EAS:1712671.1712672}.
However, the definition of quasi-cliques is problematic.
Note that a single edge already provides maximal density. 
Consequently additional objectives are needed.
One natural objective is to maximize the size of a graph with density
of $1$, however, this makes the problem equivalent to 
finding a maximal clique which, as mentioned above, is a
computationally-intractable problem~\citep{DBLP:conf/focs/Hastad96}.

\spara{Alternative definitions for density.}
Other definitions of graph density have been proposed. 
Recently, Tsourakakis proposed to measure density by counting
triangles, instead of counting edges~\citep{tsourakakis15triangle}.
Interestingly enough, it is possible to find an approximate densest
subgraph under this definition. 
An interesting future direction for our work is to study if the
decomposition proposed in this paper can be extended for the
triangle-density definition. 
Density definitions of the form $g(\abs{E}) - \alpha h(\abs{V})$, where
$g$ and $h$ are some increasing functions were studied by~\citet{DBLP:conf/kdd/TsourakakisBGGT13},
with specific focus on $h(x) = {x \choose 2}$. 
It not known whether the densest-subgraph problem according to this
definition is polynomially-time solvable or \np-hard.
Finally, a variant for $\density{X}$ adopted for directed graph,
along with polynomial-time discovery algorithm, was suggested by~\citet{khuller09dense}.
Such a definition could serve for defining decompositions of directed
graphs, which is also left for future work.

\spara{Hierarchical communities.}
A classic technique for modelling hierarchical nature of communities is with a
hierarchical blockmodel~\citep{clauset2008hierarchical}. Here we are given a tree, where the leaves are
the vertices of the original graph and each vertex in a tree is given a
probablility. We then model an edge $(u, v)$ with a probability given to the
lowest common ancestor of $u$ and $v$.
\citet{DBLP:conf/pkdd/TattiG13} studied a restricted version of this problem
where the tree yields a nested structure; inner communities being denser.
Unfortunately, no exact polynomial-time algorithm is known for the restricted or general problem.
On other hand, in the segmentation problem we based the model on degrees
and not individual edges. This allowed to us to solve the problem exactly.

\section{Experimental evaluation}
\label{sec:experiments}

We will now present our experimental evaluation.
We tested the two proposed algorithms, \decompose and \peel, 
for decomposing a graph into locally-dense subgraphs,
and we contrast the resulting decompositions against $k$-cores,
obtained with the \core algorithm.
We compare the three algorithms in terms of running time, 
decomposition size
(number of subgraphs they provide), 
and relative density of the subgraphs they return.
We also use the Kendall-$\tau$ to measure how similar are the
decompositions in terms of the order they induce on the graph
vertices.

\subsection{Experimental setup}

We performed our evaluation on 13 graphs of different sizes and
densities.
A short description of the graphs is given below, and their basic
characteristics can be found in Table~\ref{tab:basic}.

\begin{itemize}[itemsep=-1pt]
\item
{\dolphins:}
an undirected social network of frequent associations between dolphins
in a community living off Doubtful Sound in New Zealand.

\item
{\karate:}
the social network of friendships between members of a karate club at a US university in the 1970.

\item
{\lesmis:}
co-appearance of characters in Les Miserables novel by Victor Hugo.

\item
{\astro:}
a co-authorship network among arXiv Astro Physics publications. 

\item
{\enron:}
an e-mail communication network by Enron employees.

\item
{\fb:}
an ego-network obtained from Facebook.

\item
{\hepph:}
a co-authorship network among arXiv High Energy Physics publications. 

\item
{\dblp:}
a co-authorship network among computer science researchers. 

\item
{\gowalla:} a friendship network of \url{gowalla.com}.

\item
{\roadnet:} a road network of California, where 
vertices represent intersections and edges represent road segments.

\item
{\skitter:} an internet topology graph, obtained from traceroutes run daily in 2005.

\item
\airports: US flight traffic in January 2016\footnote{\url{http://www.transtats.bts.gov/}},
where vertices represent airports and weighted edges flight routes. The weights represent
the number of flights between two airports.

\item
\trains: UK train routes.\!\footnote{\url{http://data.atoc.org/}} 
The vertices represent medium or large exchange points (stations), while the weighted
edges represent scheduled routes. The weights represent the number of routes in a single week.
\end{itemize}

The first three datasets were obtained from UCIrvine 
Network Data Repository,\footnote{\url{http://networkdata.ics.uci.edu/index.php}}
and the remaining datasets, except for \airports and \trains, were obtained from Stanford SNAP Repository.\!\footnote{\url{http://snap.stanford.edu/data}}

We applied \core, \peel, and \decompose to every dataset. 
We used a computer equipped with 3GHz Intel Core i7 and 8GB of RAM.\!\footnote{The implementation is available at\\ \url{https://version.helsinki.fi/dacs}}

\begin{table}[t]
\caption{Basic characteristics of the datasets and the running times of the algorithms.}
\label{tab:basic}
\begin{tabular*}{\columnwidth}{@{\extracolsep{\fill}}l@{\hspace{2mm}}rr rrr}
\toprule
& & & \multicolumn{3}{l}{running time} \\
\cmidrule{4-6}
Name & $n$ & $m$ & \core & \peel & \decompose \\
\midrule
{\dolphins} &
62 & 159 &
1ms & 1ms & 2ms
\\
{\karate} &
34 & 78 &
1ms & 1ms & 2ms
\\
{\lesmis} &
77 & 254  &
2ms & 2ms & 3ms
\\[1mm]
{\astro} &
18\,772 & 396\,160 &
0.4s & 0.4s & 2s 
\\
{\enron} &
36\,692 & 183\,831  &
0.3s & 0.3s & 2s
\\
{\fb} &
747 & 30\,025 &
44ms & 44ms & 0.2s
\\
{\hepph} &
12\,008 & 237\,010 &
0.2s & 0.2s & 0.9s
\\[1mm]
{\dblp} &
317\,080 & 1\,049\,866 &
2s & 2s & 14s
\\
{\gowalla} &
196\,591 & 950\,327 &
2s & 2s & 9s
\\
{\roadnet} &
1\,965\,206 & 2\,766\,607 &
7s & 8s & 1m6s 
\\
{\skitter} &
1\,696\,415 & 11\,095\,298 &
21s & 21s  & 1m46s
\\[1mm]
{\airports} &
294 & 3\,995 &
11ms & 10ms & 27ms \\
{\trains} &
363 & 1\,357 &
7ms & 7ms & 23ms \\
\bottomrule
\end{tabular*}
\end{table}

\subsection{Results}

We begin by reporting the running times of the three algorithms for
all of our datasets. 
They are shown in Table~\ref{tab:basic}. 
As expected, the linear-time algorithms \core and \peel are both very fast;
the largest graph with 11 million edges and 1.7 million vertices is
processed in 21 seconds.  
However, we are also able to run the exact decomposition for all the graphs in reasonable
time, despite its running-time complexity of
$\bigo{n^2m}$. 
It takes less than 2 minutes for \decompose to process the largest
graph.  
There are three reasons that contribute to achieving this performance.  
First, we need to compute the minimum cut only $\bigo{k}$ times, where
$k$ is the number of locally-dense graphs. 
In practice, $k$ is much smaller than the number of vertices. 
Second, computing minimum cut in practice is faster than the theoretical $\bigo{nm}$ bound.
Third, as described in Section~\ref{sec:discovery}, most of the minimum cuts
are computed using subgraphs. While in theory these subgraphs can be 
as large as the original graph, in practice these subgraphs are
significantly smaller.

\begin{table}[t]
\caption{Smallest ratio of the profile function,  and the profile function of the exact solution as defined in Equation~(\ref{eq:ratio}), and the ratio of
the most inner discovered subgraph versus the actual densest subgraph.}
\label{tab:ratio}
\begin{tabular*}{\columnwidth}{@{\extracolsep{\fill}}lrrrrr}
\toprule
& \multicolumn{2}{c}{$r(\col{C}, \col{B})$} && \multicolumn{2}{c}{$\density{C_1} / \density{B_1}$}\\
\cmidrule{2-3} \cmidrule{5-6}
Name & \core & \peel && \core & \peel\\
\midrule
{\dolphins} &
0.94 & 0.83 && 0.98 & 0.98
\\
{\karate} &
0.95 & 0.99 && 0.95 & 0.99
\\
{\lesmis} &
0.86 & 0.87 && 0.96 & 1.00
\\[1mm]
{\astro} &
0.85 & 0.85 && 0.87 & 0.92
\\
{\enron} &
0.83 & 0.82 && 0.94 & 1.00
\\
{\fb} &
0.69 & 0.74 && 0.91 & 1.00
\\
{\hepph} &
0.74 & 0.75 && 1.00 & 1.00
\\[1mm]
{\dblp} &
0.80 & 0.86 && 1.00 & 1.00
\\
{\gowalla} &
0.89 & 0.92 && 0.87 & 1.00
\\
{\roadnet} &
0.81 & 0.87 && 0.84 & 0.87
\\
{\skitter} &
0.73 & 0.84 && 0.84 & 1.00
\\

\airports &
0.75 & 0.90 & & 0.93 & 1.00 \\
\trains &
0.60 & 0.84 & & 0.82 & 0.96 \\
\bottomrule
\end{tabular*}
\end{table}

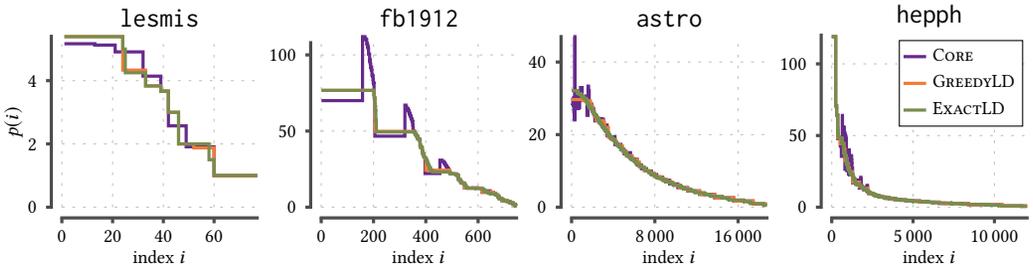
\begin{figure*}[t]
\begin{tikzpicture}
	\begin{axis}[xlabel={index $i$},ylabel= {$\prof{i}$},
		width = 2.6cm,
		cycle list name=yaf,
		scale only axis,
		title = {{\lesmis}},
		ymin = 0,
		xmin = 0,
		no markers
]
\addplot+[const plot] table[x index = 2, y index = 0, header = false] {results/lesmis_c.plot};
\addplot+[const plot] table[x index = 2, y index = 0, header = false] {results/lesmis_g.plot};
\addplot+[const plot] table[x index = 2, y index = 0, header = false] {results/lesmis_e.plot};
\pgfplotsextra{\yafdrawaxis{0}{77}{0}{5.4}}
\end{axis}
\end{tikzpicture}%
\begin{tikzpicture}
	\begin{axis}[xlabel={index $i$},
		width = 2.6cm,
		cycle list name=yaf,
		scale only axis,
		title = {{\fb}},
		ymin = 0,
		xmin = 0,
		no markers
]
\addplot+[const plot] table[x index = 2, y index = 0, header = false] {results/f1912_c.plot};
\addplot+[const plot] table[x index = 2, y index = 0, header = false] {results/f1912_g.plot};
\addplot+[const plot] table[x index = 2, y index = 0, header = false] {results/f1912_e.plot};
\pgfplotsextra{\yafdrawaxis{0}{747}{0}{108}}
\end{axis}
\end{tikzpicture}%
\begin{tikzpicture}
	\begin{axis}[xlabel={index $i$},
		width = 2.6cm,
		cycle list name=yaf,
		scale only axis,
		title = {{\astro}},
		ymin = 0,
		scaled ticks = false,
		x tick label style={/pgf/number format/1000 sep={\,}},
		xtick = {0, 8000, 16000},
		no markers,
]
\addplot+[const plot] table[x index = 2, y index = 0, header = false] {results/ca-astro_c.plot};
\addplot+[const plot] table[x index = 2, y index = 0, header = false] {results/ca-astro_g.plot};
\addplot+[const plot] table[x index = 2, y index = 0, header = false] {results/ca-astro_e.plot};
\pgfplotsextra{\yafdrawaxis{0}{18772}{0}{47}}
\end{axis}
\end{tikzpicture}%
\begin{tikzpicture}
	\begin{axis}[xlabel={index $i$},
		width = 2.6cm,
		cycle list name=yaf,
		scale only axis,
		title = {{\hepph}},
		ymin = 0,
		scaled ticks = false,
		x tick label style={/pgf/number format/1000 sep={\,}},
		no markers,
		legend entries = {\core, \peel, \decompose}
]
\addplot+[const plot] table[x index = 2, y index = 0, header = false] {results/ca-hepph_c.plot};
\addplot+[const plot] table[x index = 2, y index = 0, header = false] {results/ca-hepph_g.plot};
\addplot+[const plot] table[x index = 2, y index = 0, header = false] {results/ca-hepph_e.plot};
\pgfplotsextra{\yafdrawaxis{0}{12008}{0}{120}}
\end{axis}
\end{tikzpicture}
\caption{\label{fig:profile}Profile functions for {\lesmis}, {\fb}, {\astro}, and {\hepph}.}
\end{figure*}

Next, we compare how well \core and \peel approximate the exact
locally-dense decomposition. In order to do that we compute the ratio
\begin{equation}
\label{eq:ratio}
	r( \col{C}, \col{B}) = \min_i \frac{\prof{i; \col{C}}}{\prof{i; \col{B}}},
\end{equation}
where $\col{B}$ is the locally-dense decomposition and $\col{C}$ is
obtained by either from \peel or \core. 
These ratios are shown in Table~\ref{tab:ratio}.
We also compare $\prof{1; \col{C}}/\prof{1; \col{B}}$, that is, the ratio
of density for the  inner most subgraph in $\col{C}$ against the density of
$\col{B}_1$, the densest subgraph.
Propositions~\ref{prop:peel-approximation}
and~\ref{prop:core-approximation} guarantee that there ratios are at
least~$1/2$.  In practice, the ratios are larger, typically over $0.8$. 
In most cases, but not always, \peel obtains better ratios than
\core. 
When comparing the ratio for the inner most subgraph, \peel, by
design, will always be better or equal than \core. 
We see that only in three
datasets \core is able to find the same subgraph as \peel.

\begin{table}[t]
\caption{Sizes of the discovered decompositions and Kendall-$\tau$
  statistics between the decompositions.
\textsc{E}~stands for \decompose,
\textsc{G} for \peel, and
\textsc{C} for \core.
}
\label{tab:kendall}
\begin{tabular*}{\columnwidth}{@{\extracolsep{\fill}}lrrr rrr}
\toprule
Name & \core & \peel & \decompose & \textsc{c}-vs-\textsc{e} & \textsc{g}-vs-\textsc{e} & \textsc{c}-vs-\textsc{g} \\
\midrule
{\dolphins} &
4 & 6 & 7 &
0.76 & 0.77 & 0.99
\\
{\karate} &
4 & 3 & 4 &
0.80 & 0.95 & 0.78
\\
{\lesmis} &
8 & 8 & 9 &
0.94 & 0.99 & 0.95
\\[1mm]
{\astro} &
52 & 83 & 435  &
0.93 & 0.93 & 0.99
\\
{\enron} &
43 & 162 & 357 &
0.92 & 0.92 & 0.99
\\
{\fb} &
87 & 55 & 75  &
0.95 & 0.98 & 0.97
\\
{\hepph} &
64 & 63 & 283  &
0.93 & 0.93 & 0.98
\\[1mm]
{\dblp} &
47 & 97  & 1087 &
0.88 & 0.89 & 0.97 
\\
{\gowalla} &
51 & 161 & 899 &
0.97 & 0.96 & 0.98
\\
{\roadnet} &
3 & 43 & 2710&
0.57 & 0.80 & 0.68
\\
{\skitter} &
111 & 266 & 3501 &
0.98 & 0.97 & 0.99
\\[1mm]
\airports & 221 & 200 & 219 &
0.99 &
0.99 &
0.996 \\
\trains & 187 & 59 & 156 &
0.87 &
0.89 &
0.98 \\

\bottomrule
\end{tabular*}
\end{table}

Let us now compare the different solutions found by the three
algorithms. 
In Table~\ref{tab:kendall} we report the sizes of discovered communities and their
Kendall-$\tau$ statistics, which compares the ordering of the vertices
induced by the decompositions. 
In particular, the Kendall-$\tau$ statistic is computed by assigning each
vertex an index based on which subgraph the vertex belongs. 
To handle ties, we use the $b$-version of Kendall-$\tau$, 
as given by~\citet{agresti10ordinal}. 
If the statistic is 1, the decompositions are equal.

Our first observation is that
typically the locally-dense decomposition algorithms return more
subgraphs than the $k$-core decomposition.
As an extreme example, {\roadnet} contains only 3 $k$-cores while \peel finds 43
subgraphs and \decompose finds 2710. 
This can be explained by the fact that the vertices in the graph have
low degrees, which results in a very coarse $k$-core decomposition. 
On the other hand, \decompose and \peel exploit density to discover
more fine-grained decompositions.
This result is similar to what we presented in the
Example~\ref{ex:toy} in the introduction.

The Kendall-$\tau$ statistics are typically close to $1$, especially for large
datasets suggesting that all 3 methods result in similar decompositions.
The statistic between \core and \peel is typically larger than to
the exact solution. This is expected since \core and \peel use the exact same
order for vertices---the only difference between these two methods is how they
partition the vertex order. In addition, decompositions produced by \peel
are closer to the exact solution than the decompositions produced by \core,
which is also a natural result.

Let us now compare the solutions in terms of profile functions as
defined in Definition~\ref{definition:profile}.
We illustrate several prototypical examples of such profile functions
in Figure~\ref{fig:profile}. 
We see that \peel produces similar profiles as the
exact locally-dense decomposition. 
We also see that \core does not respect the local density constraint. 
In {\fb}, {\astro}, and {\hepph} there exist $k$-shells that are
denser than their inner shells, that is, joining these shells would increase
the density of the inner shell.  
\peel does not have this problem since by
definition it will have a monotonically decreasing profile.

\begin{figure}[h]
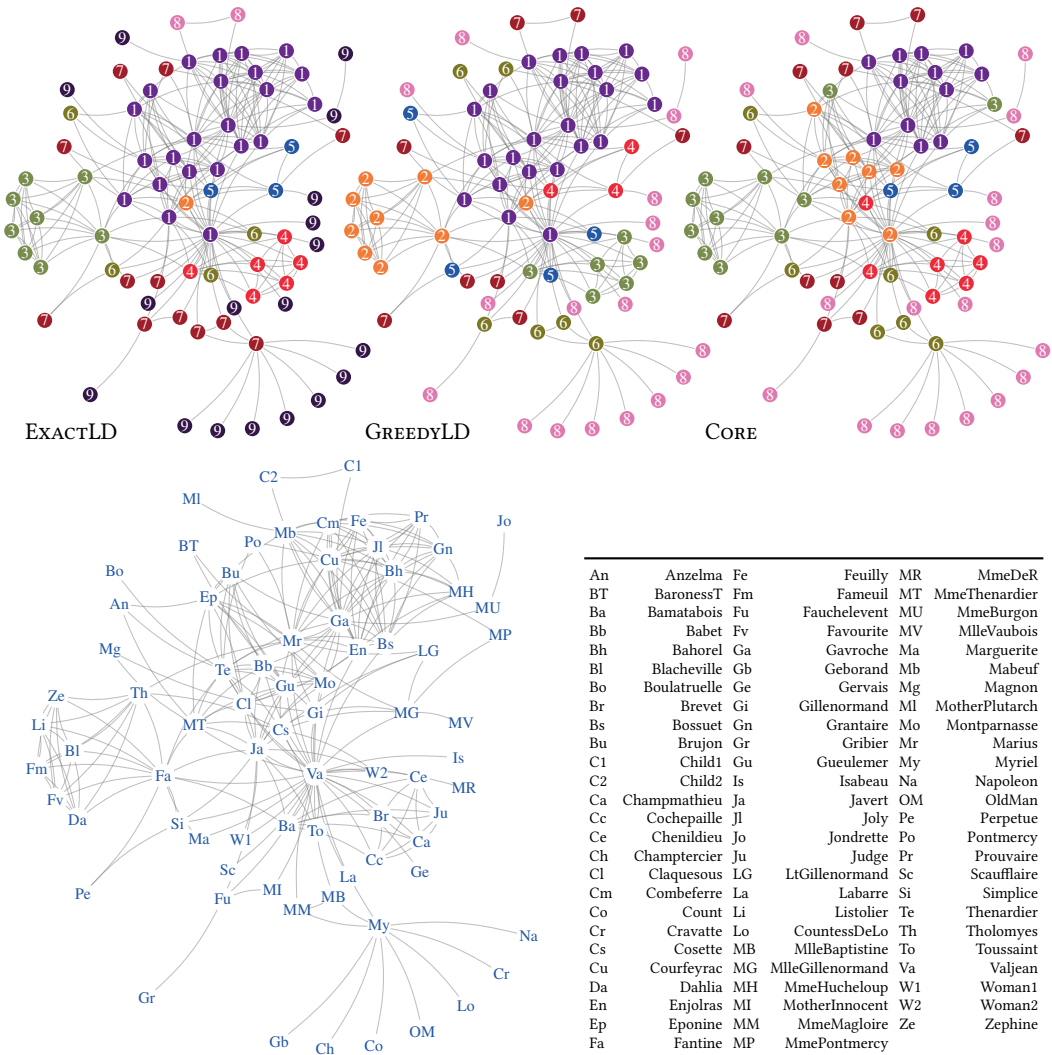


\tikzstyle{le} = [thin, gray, opacity = 0.5, bend left = 10]

\tikzstyle{lm} = [thick, fill=white, inner sep = 0pt, thick, text width = 1pt, text height = 1pt]
\tikzstyle{lc} = [circle, text width = 4pt, text height = 4pt]
\tikzstyle{lm0} = [lm, fill = yafcolor1, lc]
\tikzstyle{lm1} = [lm, fill = yafcolor2, lc]
\tikzstyle{lm2} = [lm, fill = yafcolor3, lc]
\tikzstyle{lm3} = [lm, fill = yafcolor4, lc]
\tikzstyle{lm4} = [lm, fill = yafcolor5, lc]
\tikzstyle{lm5} = [lm, fill = yafcolor6!50!black, lc]
\tikzstyle{lm6} = [lm, fill = yafcolor7, lc]
\tikzstyle{lm7} = [lm, fill = yafcolor8, lc]
\tikzstyle{lm8} = [lm, fill = yafcolor1!50!black, lc]

\tikzstyle{ll} = [font = \tiny, text = white]
\tikzstyle{ll0} = [ll]
\tikzstyle{ll1} = [ll]
\tikzstyle{ll2} = [ll]
\tikzstyle{ll3} = [ll]
\tikzstyle{ll4} = [ll]
\tikzstyle{ll5} = [ll]
\tikzstyle{ll6} = [ll]
\tikzstyle{ll7} = [ll]
\tikzstyle{ll8} = [ll]

\setlength{\tabcolsep}{-7pt}
\begin{tabular*}{\textwidth}{lll}
\begin{tikzpicture}[scale = 1]
\fontfamily{ptm}\selectfont
\input{lesmis_e}
\end{tikzpicture} &
\begin{tikzpicture}[scale = 1]
\fontfamily{ptm}\selectfont
\input{lesmis_g}
\end{tikzpicture} &
\begin{tikzpicture}[scale = 1]
\fontfamily{ptm}\selectfont
\input{lesmis_c}
\end{tikzpicture}%
\\[-5mm]
\quad\small\decompose & \quad\small\peel & \quad\small\core
\end{tabular*}

\tikzstyle{le} = [thin, gray, opacity = 0.5, bend left = 10]
\tikzstyle{lm} = [fill=white, inner sep = 0pt, circle, text width = 6pt, text height = 6pt]
\tikzstyle{ll} = [font = \tiny, text = yafcolor5]

\setlength{\tabcolsep}{0pt}
\begin{tabular*}{\textwidth}{l@{\hspace{5mm}}l}
\begin{tikzpicture}[scale = 1.4]
\fontfamily{ptm}\selectfont
\input{lesmis}
\end{tikzpicture}%
&
\begin{minipage}[b]{10cm}
\tiny
\setlength{\tabcolsep}{2pt}
\begin{tabular}{lrlrlr}
\toprule
An & Anzelma & Fe & Feuilly & MR & MmeDeR  \\
BT & BaronessT & Fm & Fameuil & MT & MmeThenardier  \\
Ba & Bamatabois & Fu & Fauchelevent & MU & MmeBurgon  \\
Bb & Babet & Fv & Favourite & MV & MlleVaubois  \\
Bh & Bahorel & Ga & Gavroche & Ma & Marguerite  \\
Bl & Blacheville & Gb & Geborand & Mb & Mabeuf  \\
Bo & Boulatruelle & Ge & Gervais & Mg & Magnon  \\
Br & Brevet & Gi & Gillenormand & Ml & MotherPlutarch  \\
Bs & Bossuet & Gn & Grantaire & Mo & Montparnasse  \\
Bu & Brujon & Gr & Gribier & Mr & Marius  \\
C1 & Child1 & Gu & Gueulemer & My & Myriel  \\
C2 & Child2 & Is & Isabeau & Na & Napoleon  \\
Ca & Champmathieu & Ja & Javert & OM & OldMan  \\
Cc & Cochepaille & Jl & Joly & Pe & Perpetue  \\
Ce & Chenildieu & Jo & Jondrette & Po & Pontmercy  \\
Ch & Champtercier & Ju & Judge & Pr & Prouvaire  \\
Cl & Claquesous & LG & LtGillenormand & Sc & Scaufflaire  \\
Cm & Combeferre & La & Labarre & Si & Simplice  \\
Co & Count & Li & Listolier & Te & Thenardier  \\
Cr & Cravatte & Lo & CountessDeLo & Th & Tholomyes  \\
Cs & Cosette & MB & MlleBaptistine & To & Toussaint  \\
Cu & Courfeyrac & MG & MlleGillenormand & Va & Valjean  \\
Da & Dahlia & MH & MmeHucheloup & W1 & Woman1  \\
En & Enjolras & MI & MotherInnocent & W2 & Woman2  \\
Ep & Eponine & MM & MmeMagloire & Ze & Zephine  \\
Fa & Fantine & MP & MmePontmercy  \\
\bottomrule
\end{tabular}\\
\end{minipage}
\end{tabular*}

\caption{Decompositions of the {\lesmis} dataset.
The upper three graphs show the decompositions of each method using numbers and colors.
The lower graph shows the abbreviated names; the table provides mapping from abbreviations
to full names.}
\label{fig:lesmis}
\end{figure}

In Figure~\ref{fig:lesmis} we present the
decompositions obtained by the three algorithms for the {\lesmis} graph.
We see that \peel obtains very similar result to the exact solution,
the only difference is the second subgraph and the third subgraph are merged and
the $7$th subgraph (in \decompose) lends vertices to the 8th last subgraph.  
While \peel has the same first subgraph as the exact solution, which is the densest
subgraph, \core breaks this subgraph into 3 subgraphs. 
Interestingly enough, the protagonist of the book, Jean Valjean, is
not placed into the first shell by \core. 

Next, we present our result with segmentation. First we computed the cost of
optimal segmentation as a function of the number of segments $k$.  Here, we
used exponential distribution as the underlying model.  The normalized scores
are shown in left plot of Figure~\ref{fig:segmentation}.  The scores behave similarly for
all datasets: they improve quickly at the very beginning (for $k = 1, \ldots,
10$), after which they settle to a relatively stable value. This value depends
on the dataset.

Next, we study how well can approximate the segmentation by using \peel instead
of the exact solution. The results are shown in the right plot of Figure~\ref{fig:segmentation}.
Here, we plot the relative difference between the approximate solution and the
optimal solution. Ideally, the difference should be 0, and
Proposition~\ref{prop:expsegment} states that it is at most 1. We see that in practice
the estimates are really close to each other: all differences are within $0.006$.
The approximation is better for smaller $k$. This is a natural result as
there is less room for disagreement in more coarse segmentations.

\begin{figure}[h]
\begin{tikzpicture}
\pgfkeys{/pgf/number format/.cd,
sci generic={mantissa sep=\times,exponent={10^{#1}}}}
\begin{axis}[xlabel={number of segments, $k$}, ylabel= {$\frac{\cost{\col{O}}}{\cost{\set{V}}}$},
        width = 3.6cm,
		height = 2.5cm,
        cycle list name=yaf,
        scale only axis,
        scaled ticks = false,
		xtick = {1, 10, 20, 30, 40},
		every axis y label/.style = {at = {(ticklabel cs:0.5)}, anchor=east, font=\scriptsize, inner sep=0pt},
        no markers,
]
\addplot table[x index = 0, y expr = {\thisrowno{2} }, header = false] {results/enron.segcost}
	node[anchor = west, inner sep = 1pt, color=black, font=\scriptsize] {\enron};
\addplot table[x index = 0, y expr = {\thisrowno{2} }, header = false] {results/as-skitter.segcost}
	node[anchor = west, inner sep = 1pt, color=black, font=\scriptsize] {\skitter};
\addplot table[x index = 0, y expr = {\thisrowno{2} }, header = false] {results/ca-hepph.segcost}
	node[anchor = west, inner sep = 1pt, color=black, font=\scriptsize] {\hepph};
\addplot table[x index = 0, y expr = {\thisrowno{2} }, header = false] {results/ca-astro.segcost}
	node[anchor = west, inner sep = 1pt, color=black, font=\scriptsize, yshift=2pt] {\astro};

\addplot table[x index = 0, y expr = {\thisrowno{2} }, header = false] {results/dblp.segcost}
	node[anchor = west, inner sep = 1pt, color=black, font=\scriptsize, yshift = -2pt] {\dblp};
\addplot table[x index = 0, y expr = {\thisrowno{2} }, header = false] {results/gowalla.segcost}
	node[anchor = west, inner sep = 1pt, color=black, font=\scriptsize] {\gowalla};
\addplot table[x index = 0, y expr = {\thisrowno{2} }, header = false] {results/f1912.segcost}
	node[anchor = west, inner sep = 1pt, color=black, font=\scriptsize, yshift=+2pt] {\fb};
\addplot table[x index = 0, y expr = {\thisrowno{2} }, header = false] {results/roadnet.segcost}
	node[anchor = west, inner sep = 1pt, color=black, font=\scriptsize] {\roadnet};

\pgfplotsextra{\yafdrawaxis{1}{40}{0.785}{1}}
\end{axis}
\end{tikzpicture}
\hfill
\begin{tikzpicture}
\pgfkeys{/pgf/number format/.cd,
sci generic={mantissa sep=\times,exponent={10^{#1}}}}
\begin{axis}[xlabel={number of segments, $k$}, ylabel= {$\frac{\cost{\col{C}} - \cost{\col{O}}}{\cost{\col{O}}}$},
        width = 3.6cm,
		height = 2.5cm,
        cycle list name=yaf,
        scale only axis,
        scaled ticks = false,
		xtick = {1, 10, 20, 30, 40},
		every axis y label/.style = {at = {(ticklabel cs:0.5)}, anchor=east, font=\scriptsize, inner sep=0pt},
        no markers,
]
\addplot table[x index = 0, y expr = {\thisrowno{4} / \thisrowno{1} - 1}, header = false] {results/enron.segcost}
	node[anchor = west, inner sep = 1pt, color=black, font=\scriptsize] {\enron};
\addplot table[x index = 0, y expr = {\thisrowno{4} / \thisrowno{1} - 1}, header = false] {results/as-skitter.segcost}
	node[anchor = west, inner sep = 1pt, color=black, font=\scriptsize] {\skitter};
\addplot table[x index = 0, y expr = {\thisrowno{4} / \thisrowno{1} - 1}, header = false] {results/ca-hepph.segcost}
	node[anchor = west, inner sep = 1pt, color=black, font=\scriptsize] {\hepph};
\addplot table[x index = 0, y expr = {\thisrowno{4} / \thisrowno{1} - 1}, header = false] {results/ca-astro.segcost}
	node[anchor = west, inner sep = 1pt, color=black, font=\scriptsize, yshift=2pt] {\astro};

\addplot table[x index = 0, y expr = {\thisrowno{4} / \thisrowno{1} - 1}, header = false] {results/dblp.segcost}
	node[anchor = west, inner sep = 1pt, color=black, font=\scriptsize, yshift = -2pt] {\dblp};
\addplot table[x index = 0, y expr = {\thisrowno{4} / \thisrowno{1} - 1}, header = false] {results/gowalla.segcost}
	node[anchor = west, inner sep = 1pt, color=black, font=\scriptsize] {\gowalla};
\addplot table[x index = 0, y expr = {\thisrowno{4} / \thisrowno{1} - 1}, header = false] {results/f1912.segcost}
	node[anchor = west, inner sep = 1pt, color=black, font=\scriptsize, yshift=-2pt] {\fb};
\addplot table[x index = 0, y expr = {\thisrowno{4} / \thisrowno{1} - 1}, header = false] {results/roadnet.segcost}
	node[anchor = west, inner sep = 1pt, color=black, font=\scriptsize, yshift=1pt] {\roadnet};

\pgfplotsextra{\yafdrawaxis{1}{40}{0}{0.006}}
\end{axis}
\end{tikzpicture}

\caption{Ratios of segmentation costs as a function of $k$ the number of segments. Here, $\col{O}$ is the optimal solution,
$\col{C}$ is the approximate solution using \peel, and $\cost{\set{V}}$ is the cost of having just one segment.
The left figure shows the optimal costs normalized by $\cost{\set{V}}$.
The right figure shows how well we approximate the exact solution by using \peel, the lower the better, ideally at 0.
\label{fig:segmentation}
}

\end{figure}
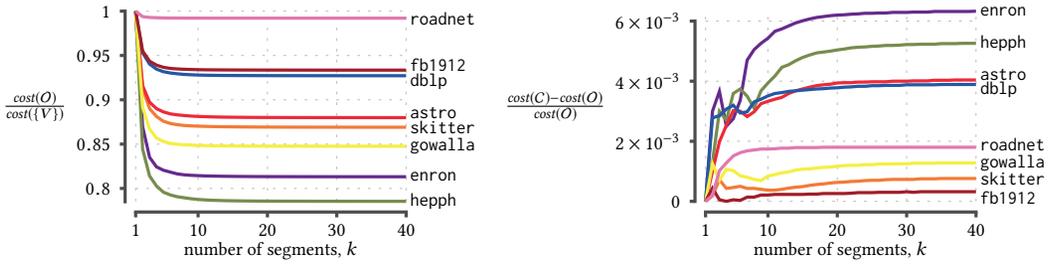

\begin{figure}[h]
\includegraphics[width = \textwidth]{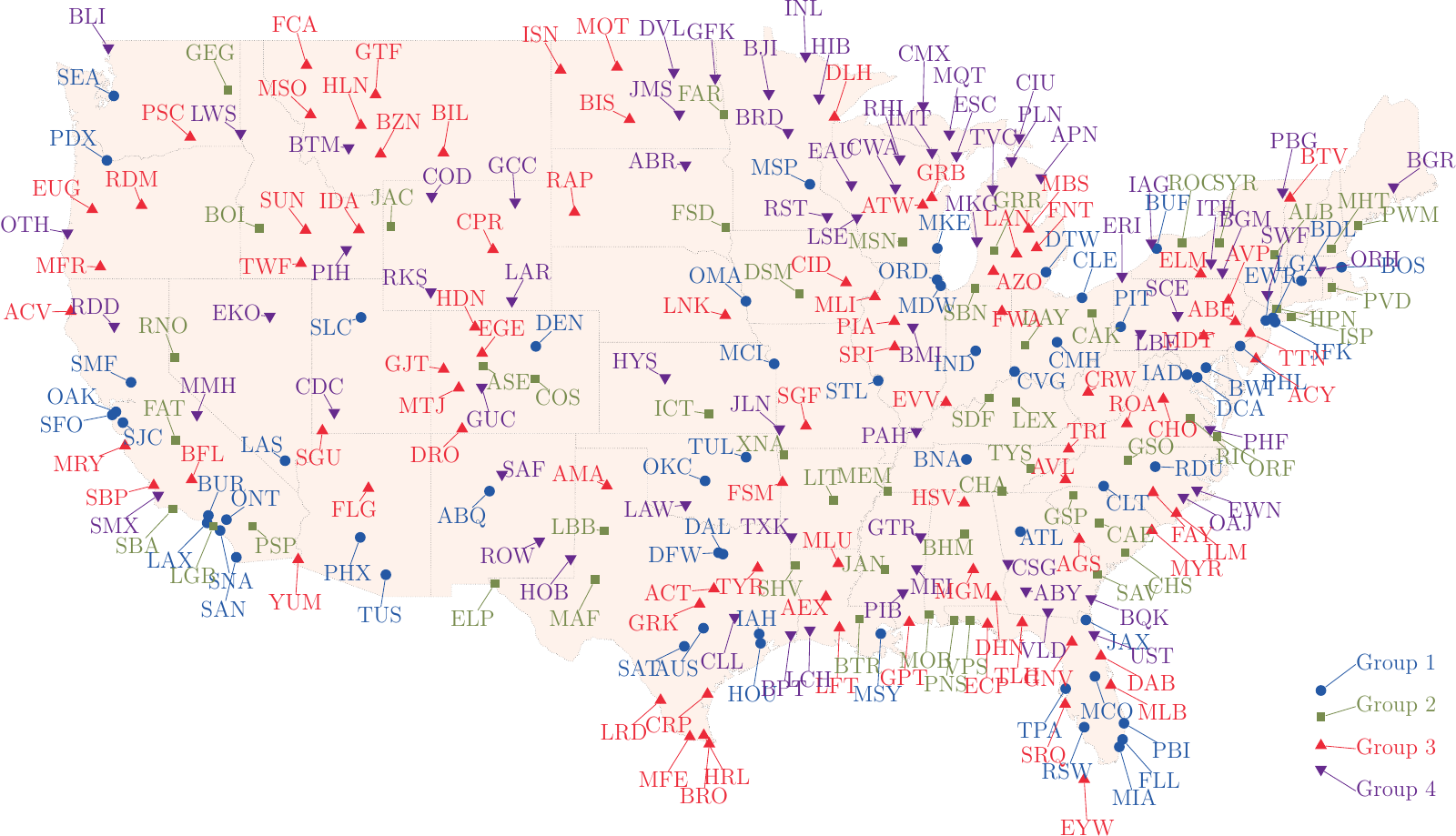}
\caption{\airports data, 4-segmentation using exponential distribution and
\decompose based on traffic data. Segments are indicated by color and shapes.
Groups with smaller indices consists of central airports that are connected
to other central airports.
}
\label{fig:flights}
\end{figure}

\begin{figure}[h]
\includegraphics[height = 17cm]{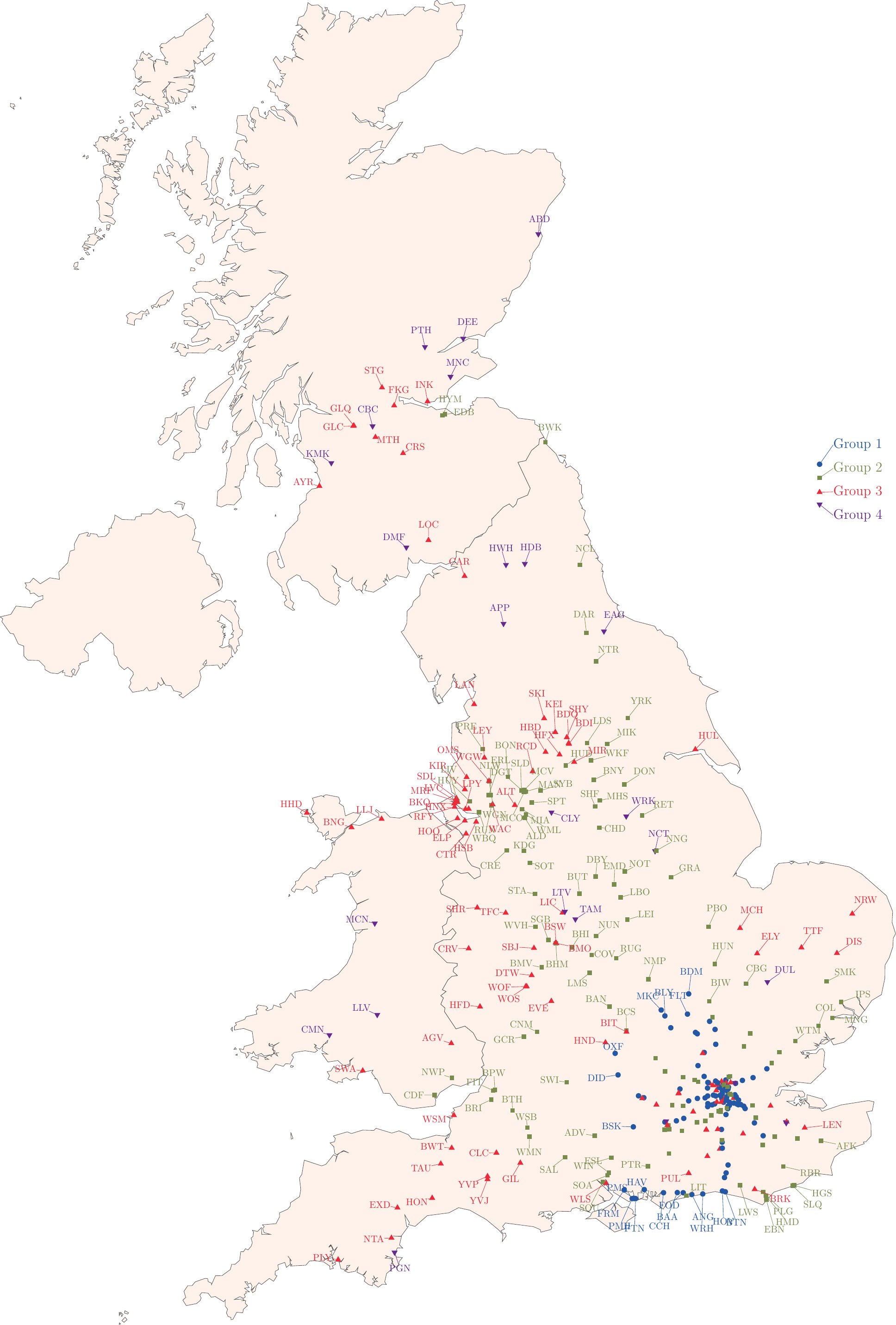}
\caption{\trains data, 4-segmentation using exponential distribution and \decompose based on traffic data. Segments are indicated by color and shapes.
To avoid clutter, London area is not annotated; see Figure~\ref{fig:london} for the zoom-in.
Groups with smaller indices consists of central stations that are connected
to other central stations.
}
\label{fig:uk}
\end{figure}

\begin{figure}[h]
\includegraphics[width = \textwidth]{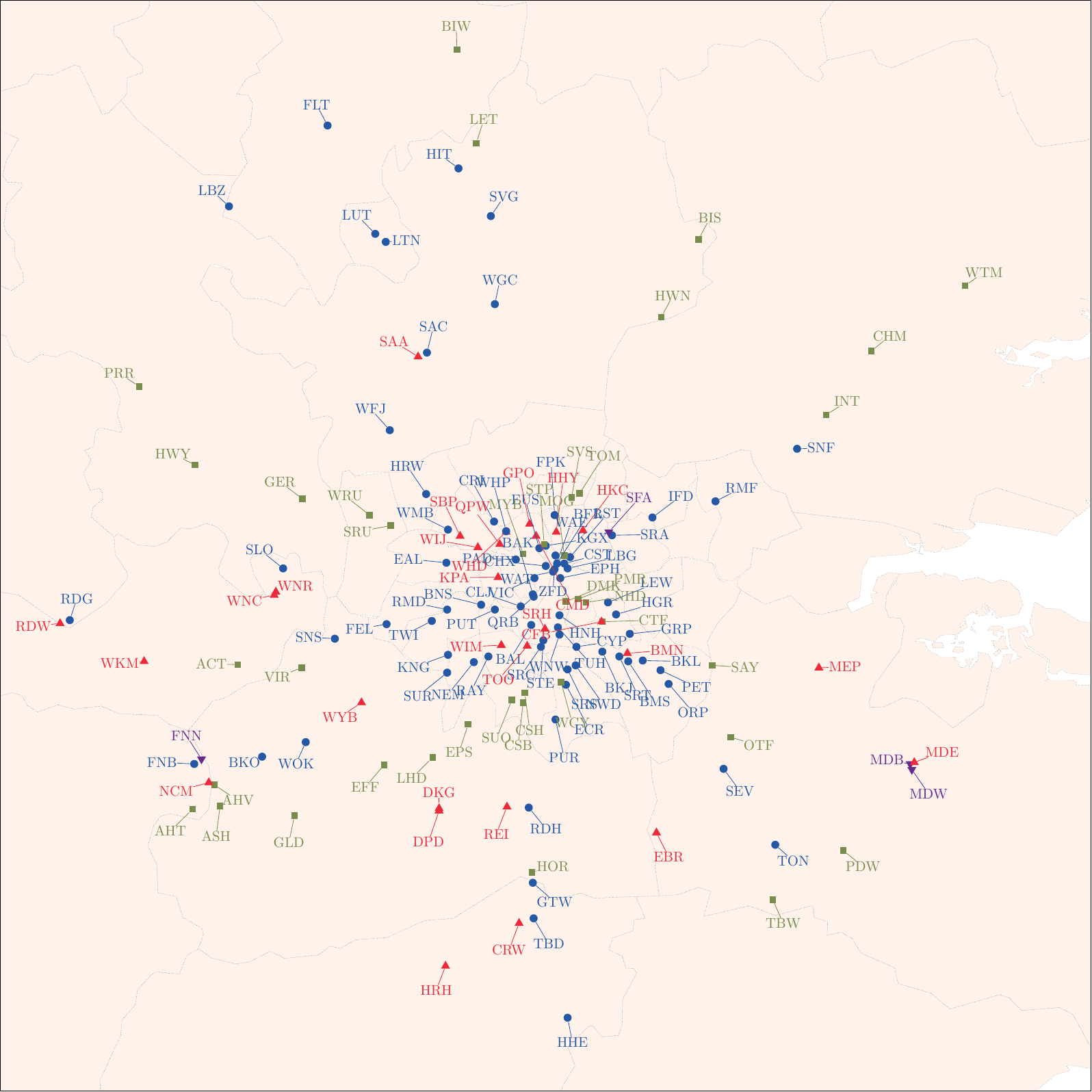}
\caption{London area zoom-in of \trains data, 4-segmentation using exponential distribution and \decompose. Segments are indicated by color and shapes.
For the whole map, see Figure~\ref{fig:uk}.
}
\label{fig:london}
\end{figure}

Finally, let us look on segmentations obtained from \trains and \airports data.
Our goal is to discover which locations, that is, train stations or airports, are central.
Here, by centrality we mean that a central location is well-connected with
others central locations. To quantify this notion we use locally-dense subgraphs.
Note that the number of locally-dense subgraphs is relatively large in these
graphs; this is due to the fact that the graphs are weighted. We were interested
to group the locations in 4 categories. So to reduce the 
the size of decomposition, we solved segmentation problem with $k = 4$ and the exponential model.
The results are shown in Figure~\ref{fig:flights}--\ref{fig:london}. 

The discovered \trains segmentation shows that the densest segment occurs in the vicinity
of London, as expected. There is also a strong concentration of the second densest segment around
Manchester/Liverpool area while the stations in Scotland, apart from the capital Edinburgh, are in outer segments.
For \airports, we see that the inner segments consists of large well-connected
airports, such as JFK, DFW, ATL, or ORD,  while the smaller, regional, airports
are assigned to the outer segments.

\section{Conclusions}
\label{sec:conclusions}

Inspired by $k$-core analysis and density-based graph mining,
we propose density-friendly graph decomposition,
a new tool for analyzing graphs. 
Like $k$-core decomposition, 
our approach decomposes a given graph into a nested sequence of
subgraphs
These subgraphs have the property that the inner subgraphs are always denser than the outer ones;
additionally the most inner subgraph is the densest one---properties that the $k$-cores do not satisfy.

We provide two efficient algorithms to discover such a decomposition. 
The first algorithm is based on minimum cut and it extends the exact
algorithm of Goldberg for the densest-subgraph problem.
The second algorithm extends a linear-time algorithm by Charikar for
approximating the same problem. 
The second algorithm runs in linear time, and thus, in addition to
finding subgraphs that respect better the density structure of the
graph, it is as efficient as the $k$-core decomposition algorithm.

In addition to offering a new alternative for decomposing a graph into
dense subgraphs, 
we significantly extend the analysis, the understanding, and the
applicability of previous well-known graph algorithms: 
Goldberg's exact algorithm and
Charikar's approximation algorithm for 
finding the densest subgraph, 
as well as the $k$-core decomposition algorithm itself.

Finally, we considered a constrained version of the problem, where
we restrict the number of subgraphs. We do this by designing a model
based on segmentation. The likelihood of this model is then optimized,
and we show that we can do this either exactly or estimate this efficiently
by a factor of 2.


%
\bibliographystyle{ACM-Reference-Format}
\bibliography{bibliography}  
\end{document}